\documentclass[a4paper,USenglish,cleveref, autoref, thm-restate]{lipics-v2021}

\usepackage[utf8]{inputenc}

\usepackage{color}
\usepackage{csquotes}
\usepackage{xspace}
\usepackage{amsmath}
\usepackage{amssymb}
\usepackage{mathtools}
\usepackage{wrapfig}
\usepackage{tikz}
\usepackage{bookmark}
\usetikzlibrary{decorations.pathmorphing}

\newcommand{\Ohtilda}{\ensuremath{\widetilde{\mathcal{O}}}\xspace}

\newcommand{\RR}{\mathbb{R}\xspace}
\newcommand{\NN}{\mathbb{N}\xspace}

\renewcommand{\epsilon}{\varepsilon}
\newcommand{\eps}{\varepsilon}

\newcommand{\tOh}{\Ohtilda}

\theoremstyle{definition}
\newtheorem{problem}[theorem]{Problem}
\newtheorem{hypo}[theorem]{Hypothesis}

\title{\texorpdfstring{Combinatorial Designs Meet Hypercliques: \\Higher Lower Bounds for Klee's Measure Problem and Related Problems in Dimensions $d\ge 4$}{Combinatorial Designs Meet Hypercliques: Higher Lower Bounds for Klee's Measure Problem and Related Problems in Dimensions d >= 4}}

\titlerunning{Higher Lower Bounds for Klee's Measure Problem and Related Problems in $\RR^d, d\ge 4$}

\author{Egor Gorbachev}{Saarbrücken Graduate School of Computer Science, Saarland Informatics Campus, Saarbrücken, Germany}{peltorator@gmail.com}{}{}

\author{Marvin Künnemann}{RPTU Kaiserslautern-Landau, Germany}{kuennemann@cs.uni-kl.de}{}{}

\authorrunning{E. Gorbachev and M. Künnemann} 

\Copyright{Egor Gorbachev and Marvin Künnemann} 

\ccsdesc[500]{Theory of computation~Design and analysis of algorithms}

\keywords{Fine-grained complexity theory, non-combinatorial lower bounds, computational geometry, clique detection} 

\category{} 

\relatedversion{} 

\acknowledgements{The second author thanks Karl Bringmann, Nick Fischer and Karol W\k{e}grzycki for helpful discussions.}

\nolinenumbers 

\newcommand{\ind}{\mathrm{ind}}
\newcommand{\ellmin}{\ell_{\mathrm{min}}}
\newcommand{\ellmax}{\ell_{\mathrm{max}}}

\EventEditors{}
\EventNoEds{0}
\EventLongTitle{}
\EventShortTitle{}
\EventAcronym{}
\EventYear{2022}
\EventDate{}
\EventLocation{}
\EventLogo{}
\SeriesVolume{}
\ArticleNo{A}

\begin{document}

\maketitle

 \begin{abstract}
	 Klee's measure problem (computing the volume of the union of $n$ axis-parallel boxes in $\mathbb{R}^d$) is well known to have $n^{\frac{d}{2}\pm o(1)}$-time algorithms (Overmars, Yap, SICOMP'91; Chan FOCS'13). Only recently, a conditional lower bound (without any restriction to ``combinatorial'' algorithms) could be shown for $d=3$ (Künnemann, FOCS'22). Can this result be extended to a tight lower bound for dimensions $d\ge 4$?

	 In this paper, we formalize the technique of the tight lower bound for $d=3$ using a combinatorial object we call \emph{prefix covering design}. We show that these designs, which are related in spirit to combinatorial designs, directly translate to conditional lower bounds for Klee's measure problem and various related problems. By devising good prefix covering designs, we give the following lower bounds for Klee's measure problem in $\mathbb{R}^d$, the depth problem for axis-parallel boxes in $\mathbb{R}^d$, the largest-volume/max-perimeter empty (anchored) box problem in $\mathbb{R}^{2d}$, and related problems:
	 \begin{itemize}
		 \item $\Omega(n^{1.90476})$ for $d=4$,
		 \item $\Omega(n^{2.22222})$ for $d=5$,
		 \item $\Omega(n^{d/3 + 2\sqrt{d}/9-o(\sqrt{d})})$ for general $d$,
	\end{itemize}
	 assuming the 3-uniform hyperclique hypothesis.
	 For Klee's measure problem and the depth problem, these bounds improve previous lower bounds of $\Omega(n^{1.777...}), \Omega(n^{2.0833...})$ and $\Omega(n^{d/3 + 1/3 + \Theta(1/d)})$ respectively.

	 Our improved prefix covering designs were obtained by (1) exploiting a computer-aided search using problem-specific insights as well as SAT solvers, and (2) showing how to transform combinatorial \emph{covering designs} known in the literature to strong prefix covering designs. In contrast, we show that our lower bounds are close to best possible using this proof technique. 
 \end{abstract}

\section{Introduction} \label{sec:introduction}

For various problems in computational geometry, the best known algorithms display a running time of the form $n^{\Theta(d)}$ where $d$ denotes the number of dimensions: Klee's measure problem and the depth problem for axis-parallel boxes in $\mathbb{R}^d$ can be solved in time $n^{d/2\pm o(1)}$~\cite{OvermarsY91, Chan10, Chan13}, a recent algorithm~\cite{Chan21} computes the largest-volume empty axis-parallel box among a given set of points in time $\tOh(n^{(5d+2)/6})$, the star discrepancy can be computed in time $O(n^{d/2+1})$~\cite{DobkinEM96}, the maximum-weight rectangle problem can be solved in time $O(n^d)$~\cite{BarbayCNP14}, to name few examples.
Indeed, for all listed problems, it can be shown~\cite{Chan10, GiannopoulosKWW12, BarbayCNP14} that an $n^{o(d)}$-time algorithm would refute the Exponential Time Hypothesis (ETH). 
Thus, the subsequent challenge is to determine running times $n^{f(d)}$ with $f(d)= \Theta(d)$ that are optimal under fine-grained complexity assumptions. By the nature of these running times (which quickly increase with $d$), it is particularly interesting to determine optimal time bounds for small dimensions such as $d\in \{2,3,4,5\}$.

For some of these problems, strong conditional lower bounds are known: For Klee's measure problem and the depth problem, Chan~\cite{Chan10} gives a tight conditional lower bound of $n^{d/2-o(1)}$ for \emph{combinatorial} algorithms -- roughly speaking, algorithms that avoid the algebraic techniques underlying fast matrix multiplication algorithms. When considering general algorithms (not only combinatorial ones), tight lower bounds are only known for \emph{weighted} problems or small dimensions: For the weighted depth problem and the maximum-weight rectangle problem, tight lower bounds of $n^{d/2-o(1)}$ and $n^{d-o(1)}$, respectively, can be shown under the Weighted $k$-Clique Hypothesis~\cite{BackursDT16}. Showing strong lower bounds for the simpler, \emph{unweighted} problems appears to be more difficult, however. For Klee's measure problem and the unweighted depth problem, a recent result shows an $n^{d/(3-3/d)-o(1)}$ conditional lower bound under the 3-uniform hyperclique hypothesis~\cite{Kunnemann22}, which yields a tight bound for $d=3$, but not for $d\ge 4$.

Thus, the motivating question of this paper is the following:
\begin{center}
\emph{Can we prove conditional optimality of known algorithms for Klee's measure problem, the depth problem and related problems for small dimensions $d\ge 4$, such as $d\in \{4,5,6\}$?}
\end{center}

\subsection{Our Results}

As a starting point of this work, we formalize the approach used in~\cite{Kunnemann22} to obtain tight hardness for $d=3$. To this end, we define the following combinatorial object, which we term \emph{prefix covering designs} (due to its conceptual similarity to certain combinatorial designs\footnote{In fact, we will later establish a formal connection between these concepts.}).

In the following definition, let $\binom{S}{t}$ denote the set of $t$-element subsets of $S$.

\begin{definition}
    Let $d, K, \alpha \in \mathbb{N}$ with $d \ge 3$ and $K \ge 4$. A \emph{$(d,K,\alpha)$-prefix covering design} consists of $d$ sequences $s_1,\dots, s_d$ over $[K]$ with the following properties.
	\begin{itemize}
    \item \textbf{Triplet condition:} For every $\{a,b,c\} \in \binom{[K]}{3}$, there are $i, i', i''\in [d]$ and $\ell,\ell', \ell''\in \mathbb{N}_0$ such that
		\begin{itemize}
			\item each element of $\{a,b,c\}$ is contained in $s_{i}[..\ell]$, $s_{i'}[..\ell']$, or $s_{i''}[.. \ell'']$. (Here, $s[.. \ell]$ denotes the prefix of the first $\ell$ elements of $s$.)
			\item $\ell + \ell' + \ell'' \le \alpha$.
		\end{itemize}
	\item \textbf{Singleton condition:} For every $x\in [K]$ occurring more than once in $s_1,\dots,s_d$, define $\ell_{\min}(x)$ ($\ell_{\max}(x)$) as the minimal (maximal) $\ell$ such that there is some $i$ with $s_i[\ell] = x$. Then we have
		\begin{itemize}
			\item $\ell_{\min}(x) + \ell_{\max}(x) \le \alpha + 1$.
		\end{itemize}
	\end{itemize}
\end{definition}

As an example, it is straightforward to see that for any $d$, the sequences $s_1 = (1,d+1), s_2 = (2, d+1), \dots, s_d = (d,d+1)$ constitute a $(d,d+1, 3)$ prefix covering design.\footnote{For the triplet condition, note that the triplet $\{a,b,c\}\in \binom{[d]}{3}$ is contained in the prefixes $s_a[.. 1], s_b[.. 1], s_c[.. 1]$ of total length $\alpha=3$ and that any triplet $\{a,b,d+1\}$ with $\{a,b\}\in \binom{[d]}{2}$ is contained in the prefixes $s_a[.. 2], s_b[.. 1]$ of total length $\alpha=3$. The singleton condition only needs to be checked for $x=d+1$, for which we note that $\ell_{\min}(d+1)=\ell_{\max}(d+1)=2$ and thus $\ell_{\min}(d+1)+\ell_{\max}(d+1)=4 \le \alpha + 1$ for $\alpha=3$.} 

Prefix covering designs constitute the core of the proof technique used in~\cite{Kunnemann22}. Specifically, we show that the existence of good prefix covering designs directly leads to strong lower bounds for several problems (these reductions are implicit in~\cite{Kunnemann22} or adapted to prefix covering designs from \cite{GiannopoulosKWW12}).
\begin{proposition}\label{prop:basic_reductions}
	Let $d, K, \alpha \in \NN$ such that there exists a $(d,K,\alpha)$ prefix covering design. Unless the 3-uniform Hyperclique Hypothesis fails, there is no $\epsilon > 0$ such that there exists an $O(n^{\frac{K}{\alpha} - \epsilon})$-time algorithm for any of the following problems:
	\begin{itemize}
		\item Klee's Measure problem in $\mathbb{R}^d$,
		\item Depth problem in $\mathbb{R}^d$,
		\item Largest-Volume Empty Anchored Box problem in $\mathbb{R}^{2d}$,
		\item Maximum-Perimeter Empty Anchored Box problem in $\mathbb{R}^{2d}$.
	\end{itemize}
\end{proposition}

Beyond these problems, similar reductions are also possible for related problems such as the Bichromatic Box problem in $\mathbb{R}^{2d}$ (given sets of red and blue points, find the axis-parallel box containing the maximum number of blue points while avoiding any red point) and various related discrepancy problems such as the Star Discrepancy, see~\cite{GiannopoulosKWW12}. Note that there is a blow-up in the dimension for the Empty Anchored Box problems, which turns out to be unavoidable assuming the 3-uniform hyperclique hypothesis, as there are $O(n^{(1/2-\epsilon)d})$-algorithms for these problems (see below). At this point, we only give a rough sketch of the reduction, with the full proof deferred to Section~\ref{sec:reductions}, where we also formally define all listed problems and discuss the 3-uniform hyperclique hypothesis. 
\begin{proof}[Proof sketch for Proposition~\ref{prop:basic_reductions}]
	For each problem, we give a reduction from the 3-uniform hyperclique problem: Given a 3-uniform hypergraph $G=(V, E)$ with $V=V^{(1)}\cup \cdots \cup V^{(K)}$ and $|V^{(1)}|=\cdots=|V^{(K)}|=n$, determine whether there are $v^{(1)}\in V^{(1)},\dots, v^{(K)}\in V^{(K)}$ that form a clique in $G$. The 3-uniform hyperclique hypothesis states that this problem requires running time $n^{K-o(1)}$.

	Intuitively, a special case of each of the problems listed above is to find an axis-parallel box $Q$ satisfying certain properties. More specifically, any candidate box $Q$ is given by choosing some value $v_i\in \{0,\dots, U-1\}$ for each dimension $i\in [d]$. We use a $(d,K,\alpha)$ prefix covering design $s_1,\dots, s_d$ to interpret the values $v_1,\dots, v_d$ as choices of vertices in $V^{(1)},\dots, V^{(K)}$: Namely, with $s_i=(s_i[1],\dots, s_i[L])$, we think of any number $v_i \in \{0,\dots, U-1\}$ with $U=n^L$ as a base-$n$ number $v_i=(v_i[1],\dots, v_i[L])$. We interpret $(v_i[1],\dots, v_i[L])\in \{0,\dots, n-1\}^L$ as choosing the $(v_i[\ell]+1)$-st vertex in $V^{(s_i[\ell])}$ for all $1\le \ell \le L$. 

	With this encoding fixed, it remains to ensure that the only true solutions $Q$ encode a \emph{clique} in $G$. This consists of two tasks: (1) ensuring that the candidate box $Q$ chooses vertices consistently, i.e., for each $V^{(x)}$ such that $x$ occurs in more than one $s_i$, we need to make sure that the same vertex is chosen in each occurrence, and (2) ensuring that the chosen vertices form a clique. Crucially, for both tasks, our geometric problems allow us to exclude candidate boxes $Q$ where the $v_i$ have certain prefixes. Specifically, due to the singleton condition, we only need to construct $O(n^\alpha)$ boxes to ensure consistency of the remaining candidate solutions $Q$. Likewise, the triplet condition is used to ensure that all candidate boxes $Q$ that encode a non-clique (for which one of the triplets $\{v^{(a)}, v^{(b)}, v^{(c)}\}$ is not an edge in $G$) are excluded, using only $O(n^\alpha)$ additional boxes. In total, this creates an instance of size $O(n^\alpha)$ for the target problem, which yields an $n^{\frac{K}{\alpha}-o(1)}$ lower bound under the 3-uniform hyperclique hypothesis.
\end{proof}

From Proposition~\ref{prop:basic_reductions}, we obtain the following direct corollary.

\begin{corollary}
	For any $d\ge 3$, let $\gamma_d \coloneqq \sup \{ \frac{K}{\alpha} \mid \text{there is a } (d,K,\alpha)\text{ prefix covering design}\}$. Then for no $\eps>0$ there exists an $O(n^{\gamma_d-\epsilon})$-algorithm for any of the problems listed in Proposition~\ref{prop:basic_reductions}, unless the 3-uniform hyperclique hypothesis fails.
\end{corollary}

The tight conditional lower bound~\cite{Kunnemann22} for Klee's Measure problem and the depth problem in $\mathbb{R}^3$ follows from the following construction: For any $g\in \NN$, we set $K=3g$, write $[K]=\{a_1,\dots, a_g, b_1,\dots,b_g, c_1,\dots, c_g\}$ and observe that \[s_1= (a_1,\dots, a_g, b_g, \dots, b_1), s_2 = (b_1,\dots,b_g, c_g, \dots, c_1), s_3 = (c_1,\dots,c_g, a_g, \dots, a_1)\]
provide a $(3, 3g, 2g+1)$ prefix covering design. Thus, we obtain $\gamma_3 \ge \lim_{g\to\infty} \frac{3g}{2g+1} = \frac{3}{2}$, establishing an $n^{\frac{3}{2}-o(1)}$ conditional lower bound for KMP in $\mathbb{R}^3$ and related problems.\footnote{It is not hard to prove that $\gamma_3 \le \frac{3}{2}$, resulting in $\gamma_3 = \frac{3}{2}$. This raises the question whether we can find exact values of $\gamma_d$ for $d\ge 4$.}

Given the direct applicability of prefix covering designs to Klee's measure problem, the depth problem and many related problems, it is only natural to ask what the highest obtainable lower bounds are using this technique. For one, designing better prefix covering designs gives stronger lower bounds. On the other hand, establishing limits for prefix covering designs may indicate potential for improved algorithms for KMP and related problems (such a phenomenon has been observed in other contexts, e.g.,~\cite{BringmannK17}).

Our first result is that prefix covering designs cannot establish a higher lower bound than $n^{\frac{d}{3} + O(\sqrt{d})}$.  The following bound will be proved in Section~\ref{sec:limits}.
\begin{proposition}\label{prop:square_root_lower_bound}
	We have that $\gamma_d \le \frac{d}{3(1-\sqrt{\frac{2}{d}})} = \frac{d}{3} + \sqrt{\frac{2}{9}}\cdot \sqrt{d} + o(\sqrt{d})$.
\end{proposition}

However, as $\frac{d}{3(1-\sqrt{2/d})} \ge \frac{d}{2}$ for $d\le 18$, this result does not rule out tight lower bounds for small dimensions. In fact, combining a computer-aided search with problem-specific insights, we give improved constructions for $d\in \{4,5\}$, which give lower bounds that are surprisingly close to $\frac{d}{2}$.

\begin{theorem} \label{thm:small_constructions}
	There is a $(4,40,21)$ prefix covering design, which yields $\gamma_4\ge \frac{40}{21} > 1.90476$. 

	There is a $(5,40,18)$ prefix covering design, which yields $\gamma_5\ge \frac{40}{18} > 2.22222$.
\end{theorem}
\begin{proof}
	The following sequences give a $(4,40,21)$ prefix covering design:
	\begin{align*}
		s_1 & = (1,	2,	3,	4,	5,	6,	7,	8,	9,	10,	40,	19,	28,	37,	26),\\
		s_2 & = (11,	12,	13,	14,	15,	16,	17,	18,	19,	20,	30,	9,	38,	27,	36),\\
		s_3 & = (21,	22,	23,	24,	25,	26,	27,	28,	29,	30,	20,	39,	8,	7,	37),\\
		s_4 & = (31,	32,	33,	34,	35,	36,	37,	38,	39,	40,	10,	29,	18,	17,	27).
	\end{align*}
	The following sequences give a $(5,40,18)$ prefix covering design:
	\begin{align*}
		s_1 &= (1,	2,	3,	4,	5,	6,	7,	8,	24,	31,	38,	30,	14),\\
		s_2 &= (9,	10,	11,	12,	13,	14,	15,	16,	32,	40,	6,	31,	22),\\
		s_3 &= (17,	18,	19,	20,	21,	22,	23,	24,	8,	7,	39,	15,	30),\\
		s_4 &= (25,	26,	27,	28,	29,	30,	31,	32,	40,	16,	23,	39,	6),\\
		s_5 &= (33,	34,	35,	36,	37,	38,	39,	40,	16,	32,	15,	23).
	\end{align*}
	For the readers' convenience, we provide checker programs to verify the singleton and triplet conditions in~\cite{ourRepo} (see Section~\ref{sec:constructions-appendix} for details).
\end{proof}

For Klee's measure problem and the depth problem in $\mathbb{R}^4$ and $\mathbb{R}^5$, the gap between the resulting conditional lower bound and the known upper bound is thus at most $O(n^{0.09524})$ and $O(n^{0.27778})$, respectively. This improves over previous hyperclique-based lower bounds of $\Omega(n^{1.777})$ and $\Omega(n^{2.0833})$, respectively. 

These results may (re-)ignite hope that it might be possible to find prefix covering designs that establish tight lower bounds for $d=4$ and $d=5$. Alas, by a careful investigation of the limits of prefix covering designs, we refute this hope.

\begin{theorem} \label{thrm:no_tight_bound}
	We have $\gamma_4 < 2$.
\end{theorem}

This result is proven via a careful analysis of the structure of prefix covering designs with quality $\frac{K}{\alpha}$ approaching 2: We show that certain \emph{levels} (i.e., $s_1[\ell],\dots, s_4[\ell]$ for certain values of $\ell$) must have a very rigid structure. Essentially, every element on such a level must have exactly a single copy on a corresponding other level. A detailed analysis of all possibilities displays a contradiction; we cannot get a quality $\frac{K}{\alpha}$ that is arbitrarily close to $2$. We give the full proof in Section~\ref{sec:limits-appendix}.
It remains an interesting question to determine the precise value of $\gamma_4$; our results yield $1.90476 \le \gamma_4 < 2$.

\paragraph*{Connection to covering designs}

Our previous results give evidence of the intricacy of designing good prefix covering designs. Unfortunately, designing optimized designs for small dimensions like $d=4$ and $d=5$ offers little insights into the asymptotics in $d$ as well as the general structure of good prefix designs for larger dimensions.

We address this by providing general constructions that are applicable for all $d$ and make use of the extensive literature on \emph{combinatorial designs}. Specifically, we observe an interesting connection between so-called \emph{covering designs} (see, e.g., the surveys~\cite{MillsM92,GordonKP95,GordonS06, coveringDesigns} and~\cite{Chan19} for an algorithmic application in computational geometry) and prefix covering designs. A $(v,k,t)$ covering design is a collection of $k$-sized subsets $B_1,\dots, B_b$ -- called \emph{blocks} -- of $[v]$ such that every $t$-element subset of $[v]$ is fully contained in some block $B_i$. These covering designs constitute a relaxation of balanced incomplete block designs.

Note that a $(d,K,\alpha)$ prefix covering design $s_1,\dots, s_d$ where each $s_i$ has length at most $L$ is superficially similar to a $(v,k,t)$-covering design with $v=K$ elements, block size $k=L$, parameter $t=3$ and $d$ blocks: in both designs, we cover triplets among $v=K$ elements using $d$ sequences/blocks. However, there are two key differences. (1) In covering designs, we cover each triplet in a single block, while in prefix covering designs, we may use prefixes from up to three sequences. (2) The sequences of prefix covering designs are inherently ordered (due to the \emph{prefix} nature of the singleton and triplet conditions), while covering designs have unordered blocks. A priori, it is unclear whether there is a general way to use good covering designs to obtain good prefix covering designs or vice versa.
Maybe surprisingly, we show how to use good $(v,k,t)$ covering designs with $t=2$ (rather than $t=3$, which might appear as the more natural correspondence) to obtain strong prefix covering designs.

Specifically, for any such covering design satisfying a mild matching-like condition (which is satisfied by many constructions known in the literature), we obtain high-quality prefix covering designs. We will see below that by plugging in known constructions, we get prefix covering designs that are close to optimal when $d\to \infty$. 

\begin{theorem} \label{thm:transformation}
	Let $d\ge 3, k\in \mathbb{N}$ and $v$ be a multiple of $d$ such that there is a $(v,k,2)$ covering design with $d$ blocks with the following property: For every block $B_i$, there exists $U_i \subseteq B_i$ of size $\frac{v}{d}$ such that $U_1,\dots, U_d$ partition $[v]$. Then $\gamma_d \ge \frac{d}{3-2\frac{v}{kd}}$.
    
\end{theorem}

Let us give an example application of this theorem (see Sections~\ref{sec:consequences} and~\ref{sec:constructions-appendix} for stronger consequences). It is well known that the projective plane of order $q$ (where $q$ is a prime power) yields a set of $v=q^2+q+1$ points, $d=q^2+q+1$ lines, with $k=q+1$ points on each line, such that every pair of points is connected by a line. This yields a $(v,k,2)$-design with $d=v=q^2+q+1$ and $k=q+1$. One can show that this design satisfies the matching-like condition (see Section~\ref{sec:limits-appendix}). Thus, for infinitely many $d$, we obtain a lower bound of $\gamma_d \ge \frac{d}{3-\frac{2}{q+1}}$. Since $q=O(\sqrt{d})$, we obtain $\gamma_d \ge \frac{d}{3-\Omega(1/\sqrt{d})} = \frac{d}{3} + \Omega(\sqrt{d})$ for infinitely many $d$, improving over the lower bound of $\gamma_d\ge \frac{d}{3} + \frac{1}{3} + \frac{1}{3(d-1)}$ that is implicit in~\cite{Kunnemann22}.

\subsection{Consequences: Improved conditional lower bounds}
\label{sec:consequences}

Using Theorem~\ref{thm:transformation}, we may take any $(v,k,2)$ covering design with $d$ blocks that is known in the literature, check whether it satisfies the matching-like condition, and obtain the corresponding lower bound on $\gamma_d$. In Table~\ref{tab:CD_results}, we list lower bounds on $\gamma_d, d\le 10$ obtained this way, specifically, by using covering designs listed in the La Jolla Covering Repository~\cite{coveringDesigns} (see Section~\ref{sec:constructions} for details). Notably, the resulting lower bounds improve over the constructions in~\cite{Kunnemann22} for $d\ge 4$.

We also provide a lower bound for \emph{all} $\gamma_d$ that is close to optimal  when $d\to \infty$.

\begin{theorem} \label{thm:general_lower_bound}
	There is some function $f(d) = d/3 + 2\sqrt{d}/9 - o(\sqrt{d})$ such that $\gamma_d \ge f(d)$ for all $d\ge 3$.
\end{theorem}

This lower bound is obtained by showing how to extend the projective planes covering designs (in a suitable way) to obtain strong prefix covering designs for all values of $d$. 

By the above theorem, we obtain a $n^{d/3 + 2/9\sqrt{d} - o(\sqrt{d})}$ conditional lower bound for Klee's measure problem and related problems. Note that Chan's reduction from $K$-clique~\cite{Chan10} can be interpreted as a lower bound of $n^{(\omega/6)d-o(1)}$ assuming that current $K$-clique algorithms are optimal. If $\omega = 2$, this cannot give any higher lower bound than $n^{d/3-o(1)}$.

Table~\ref{tab:CD_results} also lists the corresponding upper bound of $O(n^{d/2})$ for Klee's measure problem and the depth problem for comparison. The gaps for the Largest-Volume/Maximum-Perimeter Empty (Anchored) Box problem in $\mathbb{R}^d$ are a bit larger: Chan~\cite{Chan21} obtains an upper bound\footnote{While Chan focuses on the Largest-Volume Empty Box problem, he states that his algorithms for $d\ge 4$ also work for the Maximum-Perimeter version, see~\cite[Section 5]{Chan21}.} for the anchored version of $\tOh(n^{d/3 + \lfloor d/2\rfloor/6}) \le \tOh(n^{5d/12})$ for $d\ge 4$. In particular, this yields upper bounds of $\tOh(n^{2.5})$, $\tOh(n^{3.3334})$, and $\tOh(n^{4.1667})$ for $d=6$, $d=8$ and $d=10$, respectively, while we supply a conditional lower bound of $n^{\gamma_{d/2}-o(1)}$ for even $d\ge 6$, which yields lower bounds of $n^{1.5-o(1)}$, $n^{1.9047-o(1)}$ and $n^{2.2222-o(1)}$ for $d=6$, $d=8$ and $d=10$, respectively. It is an interesting question whether we can prove a higher lower bound than $n^{d/4-o(1)}$ for any $d$ or whether Chan's algorithms can be improved further.

\begin{table}
	\begin{center}
\begin{tabular}{| l | p{2cm} | p{2cm} | p{2cm} | p{2cm} |}
    \hline
    $d$ & Upper bound & Previously known lower bound & SAT-solver lower bound & Covering designs lower bound \\
    \hline
    $3$ & $1.5$ & $1.5$ &  & $1.5$\\
    \hline
    $4$ & $2$ & $1.7777$ & $1.9047$ & $1.8461$\\
    \hline
    $5$ & $2.5$ & $2.0833$ & $2.2222$ & $2.1929$\\
    \hline
    $6$ & $3$ & $2.4$ & & $2.5714$\\
    \hline
    $7$ & $3.5$ & $2.7222$ & & $3$\\
    \hline
    $8$ & $4$ & $3.0476$ & & $3.3333$\\
    \hline
    $9$ & $4.5$ & $3.375$ & & $3.6818$\\
    \hline
    $10$ & $5$ & $3.7037$ & & $4.0540$\\
    \hline
\end{tabular}
	\end{center}
    \caption{The exponents of the upper and conditional lower bounds for Klee's measure problem and the depth problem in $\mathbb{R}^d$ for $d\le 10$. The upper bound column is due to the $n^{d/2 \pm o(1)}$-time algorithms~\cite{OvermarsY91, Chan10, Chan13}, the conditional lower bounds are based on the 3-uniform hyperclique hypothesis and result from~\cite{Kunnemann22} (3rd column), Theorem~\ref{thm:small_constructions} (4th column) and from combining Theorem~\ref{thm:transformation} with covering designs found in the La Jolla Covering Repository maintained by D. Gordon~\cite{coveringDesigns} (5th column).}
		\label{tab:CD_results}
	\end{table}

\subparagraph*{Related Work}
Klee's measure problem has been well-studied since the 1970s~\cite{Klee77, Bentley77, FredmanW78, LeeuwenW81, OvermarsY91, Chan10, Chan13, Kunnemann22}, including algorithms beating $n^{d/2\pm o(1)}$ for various special cases, e.g.,~\cite{AgarwalKS07, BeumeFLPV09, Agarwal10, Bringmann12, YildizS12, Bringmann13}.

The depth problem for axis-parallel boxes is closely related to Klee's measure problem and often admits similar algorithmic ideas, see particularly~\cite{Chan13}.

Finding a largest-volume empty axis-parallel box has initially been mostly studied in two dimensions (see, e.g.,~\cite{NaamadLH84, ChazelleDL86, AggarwalS87}). In higher dimensions, Backer and Keil~\cite{BackerK10} give a $\tOh(n^d)$ algorithm, which was recently improved to $\tOh(n^{(5d+2)/6})$ by Chan~\cite{Chan21}. Note that our lower bounds are most interesting for the anchored version of the problem, which is solvable in faster running time $\tOh(n^{5d/12})$~\cite{Chan21}. Approximation algorithms have been given in~\cite{DumitrescuJ13}.
Giannopoulos et al.~\cite{GiannopoulosKWW12} give a reduction from $d$-clique, which can be understood as an $n^{(\omega/12) d-o(1)}$ lower bound assuming that current clique algorithms are optimal.

\section{Constructions} \label{sec:constructions}

In this section, we prove our general result transforming covering designs to prefix covering designs (Theorem~\ref{thm:transformation}). All remaining proofs and details on constructing prefix covering designs are deferred to Section~\ref{sec:constructions-appendix} in the appendix.

For a $(d, K, \alpha)$ prefix covering design (PCD) with sequences $s_1$, $s_2$, $\ldots$, $s_d$ we call elements $s_1[i]$, $s_2[i]$, $\ldots$, $s_d[i]$ the $i$-th level of the PCD.

When analyzing such prefix covering designs, it is helpful to distinguish between the ``first'' occurrence of some element, which we call the \emph{primary} element, and all other occurrences, which we call \emph{copies}. We call a pair $(i, \ell)$ a position if $1 \le i \le d$, $1 \le \ell$, and there exists $\ell$-th element in $s_i$.

\begin{definition}
    For any prefix covering design $s_1,\dots, s_d$, we call a position $(i, \ell)$ the \emph{primary position} of value $x$ ($1 \le x \le K$) if and only if $s_i[\ell]=x$ and $s_{i'}[\ell'] \neq x$ for every other position $(i', \ell')$ such that $(\ell', i')$ precedes $(\ell, i)$ in the lexicographic ordering.

Every other occurrence $(i', \ell')$ with $s_{i'}[\ell']=x$ is called a \emph{copy} of $x$.
\end{definition}

Note that if $(i, \ell)$ is a primary position of value $x$, then $\ell=\ell_{\min}(x)$.

\begin{definition}
    A $(v, k, t)$ covering design where $v \ge 2$, $k \ge t \ge 1$ is a collection of $k$-element subsets (called blocks) of $[v]$ such that any $t$-element subset is contained in at least one block.
\end{definition}

In the following proof, we will be extensively using $(v, k, t)$ covering designs for $t=2$. So, every pair of elements is contained in at least one block.

\begin{proof}[Proof of Theorem \ref{thm:transformation}]
    Consider some $(v, k, 2)$ covering design consisting of $d$ blocks where $v$ is divisible by $d$ and set $v' \coloneqq \frac{v}{d} \in \mathbb{N}$. Define $B_1$, $B_2$, $\ldots$, $B_d$ as the blocks of this covering design. Assume there exist sets $U_1 \subseteq B_1$, $U_2 \subseteq B_2$, $\ldots$, $U_d \subseteq B_d$ such that $|U_1| = |U_2| = \ldots = |U_d| = v'$ and $U_1, U_2, \ldots, U_d$ partition $[v]$.
    Then we will prove that for every $\varepsilon > 0$, there exist $K$ and $\alpha$ such that $\frac{K}{\alpha} \ge \frac{d}{3-2 \frac{v}{kd}} - \varepsilon$ and $(d, K, \alpha)$ PCD exists. From this we automatically get that $\gamma_d \ge \frac{d}{3 - 2 \frac{v}{kd}}$ by going to the limit.

\begin{figure}
    
\begin{center}
\begin{tabular}{| c | c c c c c c |}
    \hline
    $s_1$ & 8 & 9 & 10 & 1 & 2 & 3 \\
    \hline
    $s_2$ & 11 & 12 & 13 & 4 & 1 & 5 \\
    \hline
    $s_3$ & 14 & 15 & 16 & 7 & 1 & 6 \\
    \hline
    $s_4$ & 17 & 18 & 19 & 6 & 2 & 4 \\
    \hline
    $s_5$ & 20 & 21 & 22 & 2 & 5 & 7 \\
    \hline
    $s_6$ & 23 & 24 & 25 & 3 & 4 & 7 \\
    \hline
    $s_7$ & 26 & 27 & 28 & 5 & 3 & 6 \\
    \hline
\end{tabular}
\end{center}

\caption{Example of a ($7$, $28$, $10$) PCD construction from a $(7, 3, 2)$ covering design with $7$ blocks.}
\label{fgr:example_d_7}
\end{figure}

    First, we present a slightly worse construction.

    Order elements inside blocks of a given covering design in such a way that elements of $U_i$ are located in the first $v'$ positions of $B_i$, i.e., $\{B_i[j]\mid 1 \le j \le v'\} = U_i$.
    To construct sequences of our PCD, we take these blocks of the covering design and put $kd$ new different elements in front of them by prepending $k$ elements in each sequence. In other words, the resulting PCD has sequences $s_1$, $s_2$, $\ldots$, $s_d$ each of length $2k$ such that $s_i[j] = v + (i - 1) \cdot k + j$ for $j \le k$ and $s_i[j] = b_i[j - k]$ for $j > k$. An example for $d=7$ is given in Figure \ref{fgr:example_d_7}.
    We will prove that this gives a $(d, K, \alpha)$ PCD with $K=(v'+k)d$ and $\alpha \le 3k+v'$.

There are $v'd$ elements from a covering design and $kd$ more unique elements that we added, so $K=(v'+k)d$. It remains to check that $\alpha \le 3k+v'$.

First, we check the singleton condition. Due to our ordering of the covering design blocks, all primary positions of all elements are located in the first $k+v'$ levels, so $\ell_{\min}(x) \le k+v'$ for every element $x$. At the same time, there are $2k$ elements in each sequence in total, so $\ell_{\max}(x) \le 2k$. Thus, $\ell_{\min}(x) + \ell_{\max}(x) \le (k+v')+2k=3k+v'$ for each $x\in [K]$.

Second, we check the triplet condition. Assume we chose three elements $a$, $b$ and $c$.
Define their primary positions as $(i_a, \ell_a)$, $(i_b, \ell_b)$ and $(i_c, \ell_c)$ respectively. Without loss of generality, assume that $\ell_a \le \ell_b \le \ell_c$.
Consider two cases.

\begin{enumerate}
\item If there is at most one element from the covering design among these three, then $\ell_a \le k$, $\ell_b \le k$ and $\ell_c \le k+v'$, so we can cover them with prefixes $s_{i_a}[.. \ell_a]$, $s_{i_b}[.. \ell_b]$ and $s_{i_c}[.. \ell_c]$ of total size $\ell_a + \ell_b + \ell_c \le k + k + (k + v') = 3k+v'$.
\item If there are at least two elements from the covering design among these three, then $b$ and $c$ are in the covering design. By the definition of a covering design, there should be a sequence $s_i$ that contains both $b$ and $c$. Thus we can cover all three elements with two prefixes: $s_i[.. 2k]$ (whole sequence) and $s_{i_a}[.. \ell_a]$ of total size $2k + \ell_a \le 2k + (k+v') = 3k + v'$.
\end{enumerate}

This concludes the proof that $\alpha \le 3k+v'$ and already gives a bound $\gamma_d \ge \frac{K}{\alpha} \ge \frac{(k + v')d}{3k+v'} = \frac{d}{3} \cdot \frac{3k + 3v'}{3k+v'} = \frac{d}{3} \cdot \left(1 + \frac{2v'}{3k+v'}\right) = \frac{d}{3} \cdot \left(1 + \frac{2v}{3dk+v} \right)$.

\ 

To improve this construction we will replicate the covering design $n$ times for some positive integer $n$.
Define $B_i^j$ for $1 \le i \le d$ and $1 \le j \le n$ as the $i$-th block of the $j$-th copy of the covering design. We want different copies of the covering design to be over different elements, so the $v$ elements of $B^j$ are $\{(j-1)v+1,\dots, jv\}$. Define $U_i^j$ as $v'$-element subsets of $B_i^j$ such that $U_1^j$, $U_2^j$, $\ldots$, $U_d^j$ partition $\{(j-1)v+1,\dots, jv\}$. Define $R_i^j \coloneqq B_i^j \setminus U_i^j$ as the remaining $k-v'$ elements of each block. Also, for every sequence of our PCD, we define $m \coloneqq nk - (n-1)v'$ unique elements that are put at the beginning of this sequence. Let these unique elements for sequence $i$ be called $A_i$ ($A_i = \{nv+(i-1)m+1,\dots, nv+im\}$). Now we are ready to construct the sequences $s_1,\dots, s_d$ of our prefix covering design by
\[s_i = (A_i, U_i^1, U_i^2, \ldots, U_i^n, R_i^n, R_i^{n-1}, \ldots, R_i^1).\]

An example of such a construction is given in Figure \ref{fgr:scaled_d_7_construction}.

\begin{figure}
    
\begin{center}
\begin{tabular}{| c | c c c c c c c c c c c c c c c c |}
    \hline
$s_1$ & 22 & 23 & 24 & 25 & 26 & 27 & 28 & 1 & 8 & 15 & 19 & 20 & 12 & 13 & 5 & 6 \\
    \hline
$s_2$ & 29 & 30 & 31 & 32 & 33 & 34 & 35 & 2 & 9 & 16 & 15 & 21 & 8 & 14 & 1 & 7 \\
    \hline
$s_3$ & 36 & 37 & 38 & 39 & 40 & 41 & 42 & 3 & 10 & 17 & 15 & 18 & 8 & 11 & 1 & 4 \\
    \hline
$s_4$ & 43 & 44 & 45 & 46 & 47 & 48 & 49 & 4 & 11 & 18 & 19 & 16 & 12 & 9 & 5 & 2 \\
    \hline
$s_5$ & 50 & 51 & 52 & 53 & 54 & 55 & 56 & 5 & 12 & 19 & 21 & 17 & 14 & 10 & 7 & 3 \\
    \hline
$s_6$ & 57 & 58 & 59 & 60 & 61 & 62 & 63 & 6 & 13 & 20 & 16 & 17 & 9 & 10 & 2 & 3 \\
    \hline
$s_7$ & 64 & 65 & 66 & 67 & 68 & 69 & 70 & 7 & 14 & 21 & 20 & 18 & 13 & 11 & 6 & 4 \\
    \hline
\end{tabular}
\end{center}

\caption{Example of a $(7, 70, 24)$ prefix covering design obtained by a scaled construction with $v=7$ ($v'=1$), $k=3$, $d=7$ and $n=3$.}
\label{fgr:scaled_d_7_construction}
\end{figure}

We will prove that such a PCD has $K=(nk+v')d$ and $\alpha \le 3nk-(2n-3)v'$, similarly to the proof for the simpler construction. First, there are $v'd$ elements from every covering design, and there are $n$ designs, so overall, there are $nv'd$ elements from covering designs. Additionally, there are $md=nkd-(n-1)v'd$ more unique elements that we added, so $K=(nk+v')d$ indeed. It remains to check that $\alpha \le T$ where $T \coloneqq 3nk-(2n-3)v'$. We will use that $T = 2m+nk+v' = 3m + nv'$.

First, we check the singleton condition. Due to our ordering of the covering design blocks, all primary positions of all elements are located in the first $m+nv'$ levels, so $\ell_{\min}(x) \le m+nv'$ for every element $x$. If $\ell_{\min}(x) \le m$, this element has only one occurrence, and we do not need to check the singleton condition for it. If $\ell_{\min}(x) = m + (n - i)v' + j$ for some $1 \le i \le n$ and $1 \le j \le v'$, then it means that element $x$ belongs to the $(n-i+1)$-st covering design, and its other occurrences are located in the levels from $m+nv'+(i-1)(k-v') + 1$ to $m+nv'+i(k-v')$. So $\ell_{\max}(x) \le m+nv'+i(k-v')$. Consequently, $\ell_{\min}(x) + \ell_{\max}(x) \le (m + (n-i)v' + j) + (m+nv'+i(k-v')) = 2m + 2nv' + i(k-2v') + j \le 2m + 2nv' + n(k-2v') + v' = 2m + nk + v' = T < T + 1$ where we used the fact that $k - 2v' = k - 2 \frac{v}{d} \ge 0$ due to the lemma below. We have even proved a slightly stronger inequality:
\begin{equation} \label{eq:not-tight}
    \ell_{\min}(x) + \ell_{\max}(x) \le T.
\end{equation}

\begin{lemma}
    For every $(v, k, 2)$ covering design with $d \ge 2$ blocks, $k \ge 2 v / d$ holds.
\end{lemma}
\begin{claimproof}[Proof of Lemma]
    If $k < v$, then every element $x \in [v]$ should be located in at least two sets: otherwise, we would cover only $k-1 < v-1$ pairs involving $x$, which contradicts the fact that it is a covering design. But if every element is located in at least two sets, then the sum of all set sizes $kd$ is at least $2v$. Dividing both numbers by $d$, we get the desired inequality.

    If $k \ge v$, then $k \ge v \ge 2 \frac{v}{d}$ because $d \ge 2$.
\end{claimproof}

Second, we check the triplet condition. Consider any three elements $a$, $b$ and $c$.
Define their primary positions as $(i_a, \ell_a)$, $(i_b, \ell_b)$ and $(i_c, \ell_c)$ respectively. Without loss of generality, assume that $\ell_a \le \ell_b \le \ell_c$.
Consider two cases.

\begin{enumerate}
    \item If at most one element out of these three is from covering designs, we know that $\ell_a \le m$, $\ell_b \le m$ and $\ell_c \le m + nv'$, so we can cover them with prefixes $s_{i_a}[.. \ell_a]$, $s_{i_b}[.. \ell_b]$ and $s_{i_c}[.. \ell_c]$ with total size $\ell_a + \ell_b + \ell_c \le m + m + (m + nv') = 3m+nv'=T$.

    \item If at least two elements out of these three are from covering designs, then $b$ and $c$ are in the covering designs. By the definition of a covering design there should be a sequence $s_i$ that contains both $b$ and $c$.\footnote{$b$ and $c$ may belong to different copies of our covering design, but all copies are identical, so equivalent elements from all covering designs occur in the same sequences, so there indeed should exist such $s_i$.} Then we can cover all three elements with two prefixes: $s_i[.. \max(\ell_b^i, \ell_c^i)]$ and $s_{i_a}[.. \ell_a]$ where $\ell_b^i$ and $\ell_c^i$ are positions of elements $b$ and $c$, respectively, in the sequence $i$. We already know that elements $b$ and $c$ satisfy \eqref{eq:not-tight}. It follows that $\ell_a + \ell_b^i \le \ell_b + \ell_b^i \le T$ and $\ell_a + \ell_c^i \le \ell_c + \ell_c^i \le T$. From this we can conclude that $\ell_a + \max(\ell_b^i, \ell_c^i) \le T$, as desired.
\end{enumerate}

This concludes the proof that $\alpha \le T = 3nk-(2n-3)v'$.
This construction gives us a bound $\gamma_d \ge \frac{K}{\alpha} \ge \frac{(nk+v')d}{3nk-(2n-3)v'} = \frac{(k+ \frac{v'}{n})d}{3k-(2-\frac{3}{n})v'}$ where $n \in \mathbb{N}$ can be chosen arbitrarily. When $n$ approaches infinity, this value approaches $\frac{kd}{3k-2v'} = \frac{d}{3 - \frac{2v'}{k}} = \frac{d}{3-2\frac{v}{kd}}$. Thus, for every $\varepsilon > 0$ there exists $n$ such that such a construction gives $\frac{K}{\alpha} \ge \frac{d}{3-2 \frac{v}{kd}} - \varepsilon$, as desired.
\end{proof}

We say that a $(v, k, 2)$ covering design with $d$ blocks \emph{admits a multi-matching} if for every block $B_i$ we can choose a subset $U_i$ of size $v / d$ such that $U_1, U_2, \ldots, U_d$ partition $[v]$. The following observation shows that in Theorem \ref{thm:transformation} it is not a restriction to assume that $v$ is divisible by $d$, since we can always suitably scale covering designs:

\begin{observation} \label{obs:covering-design-scaling}
    Every $(v, k, 2)$ covering design with $d$ blocks can be transformed into a $(vd, kd, 2)$ covering design with $d$ blocks by replacing each of $v$ elements with $d$ distinct elements.
	If this scaled covering design admits a multi-matching, we get a lower bound for $\gamma_d$ of $\frac{d}{3 - 2 \frac{(vd)}{(kd)d}} = \frac{d}{3 - 2 \frac{v}{kd}}$.
\end{observation}

The bound we give depends on the existence of specific covering designs admitting multi-matchings. In Section~\ref{sec:constructions-appendix}, we transform this lower bound into a general lower bound depending only on $d$; in Table~\ref{fgr:covering-design-bounds-table}, we list the bounds we obtain for specific values of $d$. These values were obtained as follows: for a fixed value of $d$, the best lower bound can be obtained by finding a covering design that minimizes the value $\mathrm{freq} \coloneqq \frac{kd}{v}$ which we call \emph{frequency} (for a fixed covering design, the frequency is the average number of occurrences of elements). We searched for covering designs in the La Jolla Covering Repository~\cite{coveringDesigns}, fixing the number of blocks to $d$ and choosing the ones with the smallest frequencies. Then we multiplied the number of elements and set sizes in these covering designs by $d$ using Observation \ref{obs:covering-design-scaling} (because Theorem \ref{thm:transformation} works only for covering designs with $v$ divisible by $d$) and checked whether they admit multi-matching. Perhaps surprisingly,
for all specific values of $d$ that we checked, the obtained covering designs indeed admit a multi-matching. The covering designs used and their multi-matchings can be found in \cite{ourRepo} along with a computer program that checks that provided constructions are indeed covering designs, and they indeed admit multi-matchings.

The question remains whether the frequency in some dimension $d$ could be minimized by a covering design that does \emph{not} admit a multi-matching. 
Indeed, one can construct covering designs that do not admit a multi-matching.
However, since we aim to minimize the frequencies, we are considering covering designs that should have a relatively small degree of redundancy~--- otherwise, they probably could be improved. In Section~\ref{sec:constructions-appendix}, we formulate the corresponding conjecture that ``sufficiently good'' covering designs always admit a multi-matching and discuss some evidence. We also provide weaker bounds obtained from covering designs not admitting multi-matchings.

\section{Limits} \label{sec:limits}

In this section, we prove limits of prefix covering designs, i.e., upper bounds on $\gamma_d = \sup \{ \frac{K}{\alpha} \mid \text{there exists a } (d,K,\alpha) \text{ prefix covering design}\}$. The proof that $\gamma_4<2$ is deferred to Section~\ref{sec:limits-appendix} in the appendix.

The following lemma formalizes the intuition that increasing the value of $K$ can only lead to better (more precisely, not worse) prefix covering designs. 

\begin{lemma}[Scaling Lemma]\label{lemma:scaling}
	For every $(d, K, \alpha)$ PCD and positive integer $\lambda\in \mathbb{N}$, there also exists a $(d, \lambda \cdot K, \lambda \cdot \alpha)$ PCD.
\end{lemma}

The proof of this fact is deferred to Section \ref{sec:limits-appendix} in the appendix.

\begin{proof}[Proof of Proposition \ref{prop:square_root_lower_bound}]
    For a fixed $(d, K, \alpha)$ PCD define $g \coloneqq \left\lceil \frac{K}{d} \right\rceil$.
    If $\alpha \ge 3g$ then $\frac{K}{\alpha} \le \frac{K}{3g} \le \frac{K}{3K/d}=\frac{d}{3}$ and the proposition statement holds. Otherwise define $a \coloneqq \left\lceil g - \frac{\alpha}{3} \right\rceil \ge 1$, i.e., $3(g - a) \le \alpha < 3(g - a + 1)$. We will prove that $a < \sqrt{\frac{2}{d}} + 2$. If $a = 1$, it is correct, so from now on we assume that $a \ge 2$.

	Define $B$ as the set of all elements $x$ that have $\ell_{\min}(x) > g - a$. We claim that $|B|\ge d(a-1)$: The number of (not necessarily distinct) elements in the first $(g-a)$ positions (over all $s_i$) is $d \cdot (g - a) = dg - da < d \cdot (\frac{K}{d} + 1) - da = K - d(a-1)$. Since there are $K$ distinct numbers in total, the claim follows.

	To prove the proposition, we will define a graph $G_B$ with vertex set $B$. We connect two elements $x,y\in B$ by an edge if and only if there is some sequence $s_i$ containing both $x$ and $y$. We obtain our result by proving both an upper and a lower bound on the number of edges.

For a lower bound on the number of edges, consider how triplets $\{a, b, c\} \in \binom{B}{3}$ are covered by prefixes: For any such triplet $\{a,b,c\}$, there are prefixes $s_i[.. \ell]$, $s_{i'}[.. \ell']$ and $s_{i''}[.. \ell'']$ which contain $a$, $b$ and $c$ and satisfy $\ell+ \ell'+\ell'' \le \alpha$.

\begin{claim}
    Without loss of generality, we may assume that at least one of $\ell$, $\ell'$ and $\ell''$ is zero.
\end{claim}

\begin{proof}
    If all of them are at least $g - a + 1$, then $\ell+\ell'+\ell'' \ge 3(g-a+1) > \alpha$, which yields a contradiction. Otherwise, if at least one of them is at most $g-a$, then this prefix cannot contain any of $a$, $b$ and $c$ as $\ellmin(a),\ellmin(b),\ellmin(c)> g-a$. We can set this prefix to the empty prefix without loss of generality.
\end{proof}

So indeed, we can imagine that triplets of elements in $B$ must be covered by using only two prefixes, not three. In particular, for every triplet of elements from $B$, at least two of them must occur in the same sequence, i.e., they must have an edge in $G_B$. Put differently, the complement graph of $G_B$ is triangle-free and thus contains at most $|B|^2/4$ edges by Mantel's Theorem~\cite{Mantel07} (a special case of Turan's Theorem). We conclude that $G_B$ has at least $\binom{|B|}{2} - \frac{|B|^2}{4} = \frac{|B|^2 - 2 |B|}{4} \ge \frac{(|B|-2)^2}{4}$ edges because $|B| \ge 2$.

	We now show that either the number of edges is at most $dg^2/2$ or $|B|\le 2g$.  We ask on which positions elements from $B$ can be located in the sequences. We know that $\ell_{\min}(x) \ge g - a + 1$ for any $x\in B$. At the same time, if some element from $B$ is located in position $\ge 2(g-a)+3$ (in some sequence $i$), then this must be its only occurrence since otherwise, it would violate the singleton condition. Furthermore, any covering of a triplet with such an element cannot contain elements from $B$ in other sequences because it would take a prefix of length at least $2(g-a)+3$ in sequence $i$ and a prefix of length at least $g - a + 1$ in some other sequence, which would violate the triplet condition. From this, we can conclude that if every triplet with this element and other elements in $B$ is covered, all elements from $B$ have to occur in sequence $i$. We can assume that all elements have indices at most $\alpha$ (otherwise, they are useless for coverings), so there are at most $\alpha - (g - a) \le 2g$ elements from $B$ in this sequence. This yields $|B| \le 2g$.
In the remaining case all $x\in B$ satisfy $\ell_{\max}(x) \le 2(g - a) + 2$ and $\ell_{\min}(x) > (g - a)$, so there are at most $g - a + 2 \le g$ elements from $B$ in each sequence. Thus, there are at most $d \cdot \binom{g}{2}$ pairs of elements from $B$ that occur in the same sequence.

    From the above lower and upper bounds on the number of edges, we derive that
\[\frac{(|B| - 2)^2}{4} \le d \cdot \binom{g}{2} < \frac{dg^2}{2}.\]
    Combining this with the fact that $d(a-1) \le |B|$, we deduce that $d(a - 1) - 2 \le |B| - 2 < \sqrt{2d} g$. (Note that in the case $|B| \le 2g$, the upper bound is trivially satsified since $d \ge 2$.) Consequently,
    \[a < \frac{\sqrt{2d}g + 2}{d} + 1 \le \sqrt{\frac{2}{d}}g + 2\]
    for $d \ge 2$. We plug this inequality into our initial inequality on  $\alpha$:
    \[\alpha \ge 3(g-a) > 3g\left(1 - \sqrt{\frac{2}{d}} - \frac{2}{g}\right) \ge \frac{3K}{d} \left(1 - \sqrt{\frac{2}{d}} - \frac{2}{g}\right).\]
    It follows that
    \[\frac{K}{\alpha} \le \frac{K}{\frac{3K}{d} \left(1 - \sqrt{\frac{2}{d}} - \frac{2}{g}\right)} = \frac{d}{3 \cdot \left(1 - \sqrt{\frac{2}{d}} - \frac{2}{g} \right)}.\]

    Due to Scaling Lemma \ref{lemma:scaling} we know that if there exists a $(d, K, \alpha)$ PCD then there also exists a $(d, K \cdot \lambda, \alpha \cdot \lambda)$ PCD for every positive integer $\lambda$. If we plug this covering design into the inequality above, we will get that
    \[\frac{K}{\alpha} = \frac{\lambda \cdot K}{\lambda \cdot \alpha} \le \frac{d}{3 \cdot \left(1 - \sqrt{\frac{2}{d}} - \frac{2}{g'} \right)}\]
    where $g' \coloneqq \left\lceil \frac{K \cdot \lambda}{d} \right\rceil$. If we take $\lambda \to +\infty$ then $\frac{2}{g'} \to 0$ and in the limit, we get the desired upper bound on $\frac{K}{\alpha}$:
    \[\frac{K}{\alpha} \le \frac{d}{3 \cdot \left(1 - \sqrt{\frac{2}{d}} \right)} = \frac{d}{3} \cdot \left(1 + \frac{\sqrt{\frac{2}{d}}}{1 - \sqrt{\frac{2}{d}}}\right) = \frac{d}{3} + \frac{\sqrt{2d}}{3\left(1 - \sqrt{\frac{2}{d}}\right)} = \frac{d}{3} + \sqrt{\frac{2}{9}}\cdot \sqrt{d} + o(\sqrt{d}).\qedhere\]
\end{proof}

\section{Conclusion and Outlook}
\label{sec:conclusion}

In this work, we make progress on obtaining tight conditional lower bounds for Klee's measure problem and related problems for $d\ge 4$. We give improved lower bounds that leave gaps of only $O(n^{0.09524})$, $O(n^{0.27778})$ and $O(n^{0.4286})$ for $d=4$, $d=5$ and $d=6$, respectively. On the negative side, we prove that the proof technique via prefix covering designs and Proposition~\ref{prop:basic_reductions}~--- despite yielding a tight lower bound for $d=3$~--- cannot give tight lower bounds for $d\ge 4$, so that a novel reduction approach is needed for this task. Of course, it remains a tantalizing possibility that the $n^{d/2 \pm o(1)}$ running time for Klee's measure problem for large dimensions $d\ge 4$ can be broken. 

We feel that the prefix covering designs formalized in this work are interesting in their own right. We establish a connection to the well-studied covering designs, by giving a framework that turns 2-covering designs into prefix covering designs. This connection leads to the asymptotic bound $\gamma_d=\frac{d}{3} + \Theta(\sqrt{d})$, leading to an $n^{d/3 + \Theta(\sqrt{d})}$ conditional lower bound for Klee's measure problem and related problems, improving over a previous bound of $n^{d/3 + 1/3 + \Omega(1/d)}$.

\bibliography{refs}

\appendix

\section{Reductions} \label{sec:reductions}

In this section, we give the reductions from the 3-uniform hyperclique problem to all considered problems.

We first state the 3-uniform $K$-hyperclique problem, using a $K$-partite formulation that is well-known to be equivalent to the general case:

\begin{problem}[3-uniform $K$-hyperclique]
Let $K\ge 4$. We are given a $K$-partite 3-uniform hypergraph $G=(V,E)$, i.e., $V$ is the disjoint union of vertex parts $V^{(1)},\dots, V^{(K)}$, containing $n$ nodes each, and $E\subseteq \binom{V}{3}$ such that every edge connects 3 vertices from different vertex parts. The task is to determine whether there are $v_1\in V^{(1)},\dots, v_K \in V^{(K)}$ forming a $K$-clique, i.e., $\{v_i,v_j,v_k\}\in E$ for all $\{i,j,k\}\in \binom{[K]}{3}$.
\end{problem}

While the $K$-clique problem in \emph{graphs} is well known to be solvable in time $O(n^{(\omega/3)K})$ (whenever $K$ is divisible by 3), no algorithm is known that beats brute-force running time $n^{K\pm o(1)}$ for the 3-uniform $K$-hyperclique problem by a polynomial factor. In fact, it is known that the algebraic approach for the fast $K$-clique algorithms cannot work for $K$-hyperclique, see~\cite{LincolnVWW18} for a thorough discussion. 
Correspondingly, the following hypothesis has been formulated:

\begin{hypo}[3-uniform $K$-hyperclique hypothesis]
	For no $K\ge 4$ and $\epsilon > 0$, there is an $O(n^{K-\epsilon})$-time algorithm for the 3-uniform $K$-hyperclique problem.
\end{hypo}
A refutation of the 3-uniform $K$-hyperclique hypothesis would give a novel $O(2^{(1-\epsilon)n})$-algorithm for Max3SAT~\cite{AbboudBDN18, LincolnVWW18}, as well as improved algorithms for Boolean CSPs parameterized by solution size~\cite{KunnemannM20} and model-checking certain first-order properties~\cite{BringmannFK19}. It has been used to give strong conditional lower bounds for various problems, e.g.~\cite{AbboudBDN18, LincolnVWW18, BringmannFK19, VWilliamsX20, CarmeliZBKS20, AnGIJKN22, BringmannKKNP22, DalirrooyfardVW22, DalirrooyfardJVWW22}.

We turn to the first set of problems: In the \emph{coverage problem}, we receive $n$ axis-parallel boxes $A_1,\dots, A_n$ and a bounding box $B$ as input, and ask whether the union of $A_1,\dots, A_n$ is equal to $B$. Trivially, this problem is a special case of \emph{Klee's measure problem}, which receives $A_1,\dots, A_n$ as input and asks to compute the volume of the union of $A_1,\dots, A_n$. Finally, the \emph{depth problem} for unweighted axis-parallel box, receives $A_1,\dots, A_n$ as input and asks to find a point that is contained in the \emph{maximum number of boxes}, which we call the \emph{depth} of $A_1,\dots, A_n$. It is not difficult to see (cf., e.g.~\cite{Chan10}) that the coverage problem reduces to the depth problem with $n2^d$ boxes: For each $A_i$, compute a set of $2^d$ boxes that cover the complement of $A_i$ in $B$. The depth of these boxes is (at least) $n$ if and only if the original coverage problem is a NO instance.

Thus, by giving an efficient reduction to the coverage problem, we obtain the same lower bounds for Klee's measure problem and the depth problem. The following proposition is implicit in the tight conditional lower bound for Klee's measure problem and the depth problem in 3D~\cite{Kunnemann22}: 

\begin{proposition}[Implicit in~\cite{Kunnemann22}]\label{prop:coverage}
Let $d, K, \alpha \in \NN$ such that there exists a $(d,K,\alpha)$ prefix covering design. Unless the 3-uniform Hyperclique Hypothesis fails, there is no $\epsilon > 0$ such that there exists an $O(n^{\frac{K}{\alpha} - \epsilon})$-time algorithm for:
	\begin{itemize}
        \item the coverage problem for axis-parallel boxes in $\mathbb{R}^d$.
		\item the depth problem for unweighted axis-parallel boxes in $\mathbb{R}^d$.
		\item Klee's measure problem for axis-parallel boxes in $\mathbb{R}^d$.
	\end{itemize}
\end{proposition}

Before we give the proof, for technical reasons, we start with a simple observation about prefix covering designs.

\begin{observation}\label{obs:equallength}
	Let $s_1,\dots, s_d$ be a $(d,K,\alpha)$ prefix covering design. Without loss of generality, we may assume that all $s_i$ have equal length $\alpha$.
\end{observation}
\begin{claimproof}
	Let $s'_1,\dots, s'_d$ denote the prefix covering design obtained as follows: (1) we delete all elements $s_i[\ell]$ with $\ell > \alpha$ and (2) for any sequence $s'_i$ with $|s_i| <\alpha$, we set $s'_i[\ell] = s_1[1]$ for $|s_i|<\ell \le \alpha$. Note that $s'_1,\dots, s'_d$ are of equal length $\alpha$. It is straightforward to check that it remains a $(d,K,\alpha)$ prefix covering design.
\end{claimproof}

\begin{proof}[Proof of Proposition~\ref{prop:coverage}]
	Let $s_1,\dots, s_d$ be a $(d,K,\alpha)$ prefix covering design. We reduce the 3-uniform $K$-hyperclique problem to the coverage problem as follows: Let $G=(V, E)$ denote a $K$-partite 3-uniform hypergraph with vertex parts $V^{(1)},\dots, V^{(K)}$. 

    Let $L_1, \dots, L_d$ denote the lengths of the sequences $s_1, \dots, s_d$ and let $L\coloneqq \max_i L_i$. We will specify an instance of the coverage problem with bounding box $[0,U)^d$ where $U\coloneqq n^L$. For any tuple $v=(v[1],\dots,v[\ell])\in \{0,\dots, n\}^\ell$ with $\ell \le L$, we define
	\[\ind(v) = v[1] n^{L-1} + v[2] n^{L-2} + \cdots + v[\ell] n^{L-\ell}.\]
	Thus, the hypercubes $\prod_{i=1}^d [\ind(v_i), \ind(v_i)+1)$ for $v_1,\dots,v_d\in \{0,\dots, n-1\}^L$ partition the bounding box. For any subset $X\subseteq [K]$, we interpret $f:X\to\{0,\dots,n-1\}$ as choosing the $(f(x)+1)$-st vertex in $V^{(x)}$, written $v_{f(x)}^{(x)}$, for each $x\in X$. For any $v =(v[1],\dots, v[\ell]) \in \{0,\dots, n-1\}^\ell$ with $\ell \le L$, we define
	\begin{alignat*}{2}
		I_<(v) & = [\ind(v[1],\dots, v[\ell-1], 0), & &\ind(v[1],\dots,v[\ell-1],v[\ell])),\\
		I(v) & = [\ind(v[1],\dots, v[\ell-1], v[\ell]), & &  \ind(v[1],\dots,v[\ell-1],v[\ell]+1)),\\
		I_>(v) & = [\ind(v[1],\dots, v[\ell-1], v[\ell]+1), & &  \ind(v[1],\dots,v[\ell-1],n)).
	\end{alignat*}
	Note that for any $v'\in \{0,\dots, n-1\}^L$, the interval $[\ind(v'), \ind(v')+1)$ is contained in $I(v)$ if and only if $v[1],\dots, v[\ell]$ agrees with $v'[1],\dots, v'[\ell]$ (otherwise, they are disjoint). Similarly, $[\ind(v'), \ind(v')+1)$ is contained in $I_<(v)$ if and only if $v'$ and $v$ agree on the first $\ell-1$ positions \emph{and} $v'[\ell]< v[\ell]$. An analogous statement holds for $I_>(v)$. For the empty prefix $v=()$, we define $I(v) = [0,n^L)$.

	We will construct boxes of the following form: For any distinct dimensions $i_1,\dots, i_k \in [d]$ and intervals $I_1,\dots, I_k$ in $[0,U)$, we define the \emph{checking box} $B(i_1:I_1, \dots, i_k:I_k) = J_1\times \cdots\times J_d$  where $J_{i_1} = I_1,\dots, J_{i_k} = I_k$ and $J_i = [0,U)$ for all $i\in [d]\setminus \{i_1,\dots, i_k\}$.

We define two types of boxes. The first type are \emph{edge-checking boxes}. Specifically, for any $\{a,b,c\} \in \binom{[K]}{3}$, we will construct a box as follows: According to the triplet condition of the prefix covering design $s_1,\dots, s_d$, there are $i, i', i''\in [d]$ and $\ell,\ell',\ell''\in \mathbb{N}_0$ such that the set $X$ of elements occurring in $s_i[..\ell], s_{i'}[..\ell'], s_{i''}[..\ell'']$ satisfies $a,b,c\in X$ and $|X|\le \alpha$. For any function $f: X \to \{0,\dots, n-1\}$ with $\{v_{f(a)}^{(a)},v_{f(b)}^{(b)}, v_{f(c)}^{(c)}\}\notin E$, we define the edge-checking box
	\begin{align*}
		C_{a,b,c,f} = B(& i: I(f(s_i[1]),\dots, f(s_i[\ell])),\\
		  & i': I(f(s_{i'}[1]),\dots, f(s_{i'}[\ell'])),\\
		  & i'': I(f(s_{i''}[1]),\dots, f(s_{i''}[\ell'']))).
	\end{align*}

	The second type of boxes are \emph{consistency-checking boxes}. Specifically, for any $x\in [K]$ that occurs more than once in $s_1,\dots, s_d$ we do the following. Let $(i,\ellmin)$ be such that $\ellmin$ is minimal with $s_i[\ellmin] = x$. For every $(i', \ell')\ne (i,\ellmin)$ with $s_{i'}[\ell']=x$, we let $X_{i',\ell'}$ denote the elements occurring in $s_i[..\ellmin], s_{i'}[..\ell']$. By singleton condition, we have $|X_{i',\ell'}|\le \alpha$, since these prefix have total length $\ellmin+\ell' \le \ellmin(x) + \ellmax(x) \le \alpha+1$ and $x$ occurs in both of them. For each $f:X\to\{0,\dots, n-1\}$, we construct the consistency-checking boxes
	\begin{align*}
		C^<_{x, i', \ell', f} = B( & i: I(f(s_i[1]),\dots,f(s_i[\ellmin-1]), f(x)) \\
                & i': I_<(f(s_{i'}[1]),\dots, f(s_{i'}[\ell'-1]), f(x))),\\
		C^>_{x, i', \ell', f} = B( & i: I(f(s_i[1]),\dots,f(s_i[\ellmin-1]), f(x)) \\
                & i': I_>(f(s_{i'}[1]),\dots, f(s_{i'}[\ell'-1]), f(x))).
	\end{align*}
	This finishes the construction of the reduction.

    Recall that hypercubes $\prod_{i=1}^d [\ind(v_i), \ind(v_i)+1)$ for $v_1,\dots,v_d\in \{0,\dots, n-1\}^L$ partition the bounding box. We say that $v_1,\dots, v_d$ is a \emph{consistent encoding}, if for each $(i,\ell),(i',\ell')$ with $s_i[\ell] = s_{i'}[\ell']$, we have $v_i[\ell] = v_{i'}[\ell']$. Note that there is a bijection between consistent encodings $v_1,\dots, v_d$ and $v^{(1)}\in V^{(1)}, \dots,v^{(K)}\in V^{(K)}$ given by setting $v^{(x)}=v^{(x)}_{v_i[\ell]}$ for each $(i,\ell)$ with $s_i[\ell] = x$.
	
	We claim that any hypercube $\prod_{i=1}^d [\ind(v_i), \ind(v_i)+1)$ for $v_1,\dots,v_d\in \{0,\dots, n-1\}^L$ remains uncovered by the constructed boxes if and only if $v_1,\dots, v_d$ gives a consistent encoding of $K$-hyperclique in $G$: 
	\begin{enumerate}
        \item Any inconsistent encoding is covered by a corresponding consistency-checking box: Specifically, let $x\in [K], i',i''\in [d]$ and $\ell'\ge \ell''$ be such that $s_{i'}[\ell']=s_{i''}[\ell'']=x$, $v_{i'}[\ell'] \ne v_{i''}[\ell'']$ and $\ell'$ is minimal with this property, i.e., all positions $(i,\ell)$ with $\ell < \ell'$ are consistent. Then, let $(i,\ellmin)$ be as chosen in the definition of the consistency-checking boxes for $x$. Note that either $v_{i'}[\ell']\ne v_i[\ellmin]$ or $v_{i''}[\ell'']\ne v_i[\ellmin]$; without loss of generality, assume that $v_{i'}[\ell']\ne v_i[\ellmin]$. Recall that $X_{i',\ell'}$ denotes the set of elements occurring in $s_i[..\ellmin], s_{i'}[..\ell']$ and note that $v_i[..\ellmin], v_{i'}[..\ell'-1]$ is a consistent encoding of some $f:X_{i',\ell'}\to \{0,\dots, n-1\}$.

			Conversely, it is straightforward to see that the hypercubes of all consistent encodings are left uncovered by the consistency-checking boxes.

    \item The consistent encoding of any $v^{(1)}\in V^{(1)}, \dots, v^{(K)}\in V^{(K)}$ that does \emph{not} form a hyperclique is covered by a corresponding edge-checking box: Let $\{a,b,c\} \in \binom{[K]}{3}$ such that $\{v^{(a)},v^{(b)}, v^{(c)}\}\notin E$, witnessing that $v^{(1)}, \dots, v^{(K)}$ do not form a hyperclique. We claim that the corresponding consistent encoding $v_1,\dots, v_d$ is covered by the edge-checking box $C_{a,b,c,f}$ for an appropriately chosen $f$. Specifically, let $s_i[..\ell], s_{i'}[..\ell'], s_{i''}[..\ell'']$ denote the prefixes chosen according to the triplet condition, let $X$ denote the set of elements occurring in these prefixes and let $f:X\to\{0,\dots, n-1\}$ be defined by $f(s_j[t])=v_j[t]$ for all $(j,t) \in \{(i,1),\dots,(i,\ell),(i',1),\dots, (i',\ell'), (i'',1), \dots, (i'',\ell'')\}$ (this is well-defined for the consistent encoding $v_1,\dots, v_d$). Then the hypercube $\prod_{i=1}^d[\ind(v_i),\ind(v_i)+1)$ is covered by this edge-checking box $C_{a,b,c,f}$, which exists since $\{v^{(a)}, v^{(b)}, v^{(c)}\}\notin E$. 

			Conversely, it is straightforward to see that the only consistent encodings covered by some edge-checking boxes $C_{a,b,c,f}$ correspond to vertex choices $v^{(1)}, \dots, v^{(K)}$ such that $\{v^{(a)},v^{(b)},v^{(c)}\}\notin E$.  
	\end{enumerate}

Finally, we obtain our conditional lower bound: The number of constructed edge-checking boxes $C_{a,b,c,f}$ is at most $\binom{K}{3} n^\alpha = O(n^\alpha)$ as there are at most $n^\alpha$ functions of the form $f:X\to\{0,\dots, n-1\}$ for $|X|\le \alpha$. Similarly, there are at most $O(n^\alpha)$ consistency-checking boxes $C^<_{x,i',\ell', f}$ and $C^<_{x,i',\ell', f}$, as there are $O(1)$ choices for $x, i',\ell'$ and at most $n^\alpha$ functions of the form $f:X\to\{0,\dots, n-1\}$ for $|X|\le \alpha$.  These boxes can be constructed in time $O(n^\alpha)$. Thus, any $O(N^{\frac{K}{\alpha}-\epsilon})$-time algorithm for the coverage problem on $N=O(n^\alpha)$ boxes would decide the given 3-uniform $K$-hyperclique instance in time $O(n^{K-\alpha\epsilon})$, which would refute the 3-uniform hyperclique conjecture.

	By the known reductions from the coverage problem to Klee's measure problem and the depth problem, the proposition follows.
\end{proof}

Note that for Klee's Measure Problem, the above reduction in fact shows a $n^{K/\alpha -o(1)}$ lower bound already under the weaker assumption that \emph{counting} $K$-hypercliques in 3-uniform hypergraphs requires time $n^{K-o(1)}$.

\paragraph*{Empty Box problems}
We turn to problems of finding ``large'' axis-parallel boxes containing no points. Specifically, in the Largest-Volume Empty Box problem, we are given a set $P$ of $n$ points in $\mathbb{R}^d$ and a bounding box $B_0=[0,U]^d$. The task is to find the largest-volume box $B$ in $B_0$ that contains no point of $P$ in its interior. We study the simpler \emph{anchored} variant, in which the box must have the form $(0,b_1)\times (0,b_2)\times \cdots (0,b_d)$, i.e., it has the origin as one of its vertices.

Likewise, in the Maximum-Perimeter Empty Anchored Box problem, on the same input, the task is to find the box $B$ of the form $(0,b_1)\times (0,b_2)\times \cdots (0,b_d)\subseteq B_0$ containing no points of $P$ in its interior and maximizing $b_1+\cdots+b_d$.

We generalize the ideas from \cite{GiannopoulosKWW12} for similar problems to exploit prefix covering designs:

\begin{proposition}
	Let $d, K, \alpha \in \NN$ such that there exists a $(d,K,\alpha)$ prefix covering design. Unless the 3-uniform Hyperclique Hypothesis fails, there is no $\epsilon > 0$ such that there exists an $O(n^{\frac{K}{\alpha} - \epsilon})$-time algorithm for:
	\begin{itemize}
		\item the Maximum-Perimeter Empty Anchored Box problem in $\mathbb{R}^{2d}$. 
		\item the Largest-Volume Empty Anchored Box problem in $\mathbb{R}^{2d}$. 
	\end{itemize}
\end{proposition}
\begin{proof}
	We give a proof analogous to Proposition~\ref{prop:coverage}, using ideas of the reduction from $d$-clique to bichromatic rectangle given in~\cite{GiannopoulosKWW12}. We give the arguments for the Max-Perimeter Empty Anchored Problem first, and describe the changes for the Largest-Volume Empty Anchored Box Problem later.

	As in~\cite{GiannopoulosKWW12}, we view $\RR^{2d}$ as the product of the two-dimensional subspaces 
	\[\RR_i^2 = \{(x_1,y_1,\dots, x_d, y_d)\mid x_j = y_j = 0, j\ne i)\} \subseteq \RR^{2d}.\]
	For a point $p \in \RR^{2d}$ and $i\in [d]$, we let $p_i$ denote its projection to $\RR_i^{2}$.

    Let $s_1,\dots, s_d$ be a $(d,K,\alpha)$ prefix covering design. We reduce the 3-uniform $K$-hyperclique problem to the Maximum-Perimeter Empty Anchored Box problem as follows: Let $G=(V, E)$ denote a $K$-partite 3-uniform hypergraph with vertex parts $V^{(1)},\dots, V^{(K)}$. 

	By Observation~\ref{obs:equallength}, we may assume that the sequences $s_1, \dots, s_d$ have equal length $L\coloneqq \alpha$.  
	As in Proposition~\ref{prop:coverage}, for any tuple $v=(v[1],\dots,v[\ell])\in \{0,\dots, n\}^\ell$ with $\ell \le L$, we define
		\[\ind(v) = v[1] n^{L-1} + v[2] n^{L-2} + \cdots + v[\ell] n^{L-\ell}.\]
	Let $U \coloneqq n^L$. To represent admissible choices of $\{0, \dots, U-1\}$ in $\RR_i^2$, we use the scaffold construction of~\cite{GiannopoulosKWW12}, i.e., we define for every $x\in \{0, \dots, U\}$ the point $p(x;i)\in \RR_i^2$ by 
	\[ p_i(x;i) = (x,U-x). \] 

    For an anchored box $X$, let $b$ denote the vertex of $X$ opposite the origin. Note that if $X$ contains no point $p(x;i)$, then the contribution of $\RR^{2}_i$ to the perimeter is at most $U+1$, which is attained, by choosing $b$ such that $b_i = (x+1, U-x)$ for some $x\in \{0,\dots, U-1\}$. Furthermore, the contribution of $\RR^{2}_i$ to the perimeter is $U+1$ only if $(x, U-x-1)< b_i$ for some $x\in \{0,\dots, U-1\}$.

	Thus, any empty anchored box of perimeter $\ge d(2U+2)$ chooses, for each $i\in [d]$, some $x_i\in \{0,\dots, U-1\}$ with $(x_i,U-x_i-1) < b_i$. We now view $x_i$ as $x_i\in \{0,\dots, n-1\}^{L}$ and introduce additional points such that the corresponding choice gives an empty box if and only if $x_1,\dots,x_d$ encode a $K$-hyperclique in $G$ (according to the prefix covering design $s_1,\dots, s_d$).

    As a gadget, we show how to introduce points that exclude certain intervals $I\subseteq \{0,\dots, U-1\}$. Specifically, for any $I\subseteq \{0,\dots, U-1\}$
	and $i\in [d]$, we define $q(I;i)\in \RR^{2}_i$ by 
	\[ q_i(I;i) = (\min(I), U-\max(I)-1). \]

	For any distinct dimensions $i_1,\dots, i_k \in [d]$ and intervals $I_1, \dots, I_k$ in $\{0,\dots, U-1\}$, we can define $q(i_1:I_1, \dots, i_k:I_k)$ as \emph{checking point} analogous to the \emph{checking box} $B(i_1:I_1, \dots, i_d:I_d)$ as 
	\[ q(i_1:I_1, \dots, i_d:I_d) = q(I_1;i_1) + \cdots + q(I_d;i_d). \]
	Note that an empty box of perimeter $\ge d(U+1)$ that does not contain $q(i_1:I_1, \dots, i_d:I_d)$ must choose, for each $i\in [d]$, some $x_i\in \{0,\dots, U-1\}$, such that for all $j\in [k]$, $x_k\notin I_k$.
	This gives precisely the same semantics as the edge-checking boxes in Proposition~\ref{prop:coverage}. Thus, we can define the checking objects $C_{a,b,c,f}$, $C_{x,i',\ell',f}^<$, and $C_{x,i',\ell',f}^>$ as in Proposition~\ref{prop:coverage} (by replacing $B(\cdot)$ by $q(\cdot)$).

	Note that we obtain $O(n^{\alpha})$ checking points $q(\cdot)$ (as analyzed in Propostion~\ref{prop:coverage}), as well as an additional $O(n^\alpha)$ points $p(x;i)$ for $x\in \{0,\dots, U-1\}, i\in [d]$, since $U=n^L = n^\alpha$.
	We thus obtain that a $O(N^{\frac{K}{\alpha}-\epsilon})$-time algorithm for the Maximum-Perimeter Empty Anchored Box problem on $N=n^\alpha$ points would give a $O(n^{K-\alpha \epsilon})$-time algorithm for the 3-uniform $K$-hyperclique problem, refuting the 3-uniform hyperclique conjecture.

	The above reduction can be adapted to the Largest-Volume Empty Anchored Box problem in a straightforward way, analogously to the reduction to the Maximum Empty Star problem in~\cite{GiannopoulosKWW12}. Namely, we define $C=\mu^U$ for some $\mu$ to be determined later and replace the definition of each point $p(x;i)\in \RR_i^2$ by 
	\[ p_i(x;i) = (\mu^i,\frac{C}{\mu^i}). \] 
	This way, the largest-volume anchored box in $[0,C]^{2d}$ not containing any point $p(x;i)$ chooses a vertex $b_i\in \RR_i^2$ with $b_i = (\mu^{x_i+1}, C/\mu^{x_i})$, and has volume $(C\mu)^d$. Furthermore, we replace the points $q(I;i)\in \RR^{2}_i$ by 
	\[ q_i(I;i) = (\mu^{\min(I)}, \frac{C}{\mu^{\max(I)-1}}). \]
	These point definitions lead to the same semantics as for the Maximum-Perimeter Problem, where instead of an empty anchored box of perimeter $d(U+1)$, we now search for an empty anchored box of volume $(C\mu)^d$. We thus obtain a $n^{\frac{K}{\alpha}-o(1)}$ lower bound under the 3-uniform hyperclique conjecture analogously to before. 
\end{proof}

\section{Constructions: Omitted Proofs and Details} \label{sec:constructions-appendix}

\begin{figure}

\begin{center}
    \begin{tabular}{| l | p{2cm} | p{2cm} | p{2cm} | p{2cm} | p{2cm} |}
    \hline
    $d$ & Best known upper bound from \cite{Chan13} & Previously known lower bound from \cite{Kunnemann22} & SAT-solver lower bound from Theorem \ref{thm:small_constructions} & Covering designs lower bound from Theorem \ref{thm:transformation} & ($v$, $k$) of the covering design \\
    \hline
    $3$ & $1.5$ & $1.5$ &  & $1.5$ & ($3$, $2$)\\
    \hline
    $4$ & $2$ & $1.777$ & $1.9047$ & $1.8461$ & ($20$, $12$)\\
    \hline
    $5$ & $2.5$ & $2.0833$ & $2.2222$ & $2.1929$ & ($45$, $25$)\\
    \hline
    $6$ & $3$ & $2.4$ & & $2.5714$ & ($6$, $3$)\\
    \hline
    $7$ & $3.5$ & $2.7222$ & & $3$ & ($7$, $3$)\\
    \hline
    $8$ & $4$ & $3.0476$ & & $3.3333$ & ($24$, $10$)\\
    \hline
    $9$ & $4.5$ & $3.375$ & & $3.6818$ & ($90$, $36$)\\
    \hline
    $10$ & $5$ & $3.7037$ & & $4.0540$ & ($80$, $30$)\\
    \hline
    $11$ & $5.5$ & $4.0333$ & & $4.4160$ & ($308$, $110$)\\
    \hline
    $12$ & $6$ & $4.3636$ & & $4.8$ & ($36$, $12$)\\
    \hline
    $13$ & $6.5$ & $4.6944$ & & $5.2$ & ($13$, $4$)\\
    \hline
    $14$ & $7$ & $5.0256$ & & $5.5322$ & ($966$, $294$)\\
    \hline
    $15$ & $7.5$ & $5.3571$ & & $5.8823$ & ($405$, $120$)\\
    \hline
    $16$ & $8$ & $5.6888$ & & $6.2249$ & ($880$, $256$)\\
    \hline
    $17$ & $8.5$ & $6.0208$ & & $6.5796$ & ($782$, $221$)\\
    \hline
    $18$ & $9$ & $6.3529$ & & $6.9428$ & ($198$, $54$)\\
    \hline
    $19$ & $9.5$ & $6.6851$ & & $7.3076$ & ($19$, $5$)\\
    \hline
    $20$ & $10$ & $7.0175$ & & $7.6923$ & ($80$, $20$)\\
    \hline
    $21$ & $10.5$ & $7.35$ & & $8.0769$ & ($21$, $5$)\\
    \hline   
\end{tabular}

\end{center}

\caption{Table of the exponents of lower bounds for Klee's measure problem and the depth problem in $\RR^d$ for $d \le 21$ acquired via covering designs and its comparison to other lower and upper bounds. All the covering designs used and a verifier program can be found in \cite{ourRepo}.}

\label{fgr:covering-design-bounds-table}

\end{figure}

The constructions for $d=4$ and $d=5$ discussed below were obtained by assuming that $K = gd$ for some positive integer $g$, fitting primary occurrences of all elements in the first $g$ levels of sequences, i.e., $s_i[j] = (i - 1)g + j$ for all $1 \le i \le d$, $1 \le j \le g$, and then running a SAT-solver for fixed values of $d$, $g$ and $\alpha$ that tries to fill levels from $g+1$-st to determine whether such PCDs exists. There is no guarantee though that these constructions are optimal. For $d \ge 6$ the constructions obtained by this SAT-solver were not better than the results obtained by Theorem \ref{thm:transformation} for a general case, so we only discuss cases of $d=4$ and $d=5$ here.

\begin{proof}[Proof of Theorem \ref{thm:small_constructions}]

    We prove these two facts by presenting constructions of prefix covering designs with desired values of $\frac{K}{\alpha}$.

    For $d=4$ we present the following ($4$, $40$, $21$) PCD which gives us $\frac{K}{\alpha} = \frac{40}{21} > 1.90476$.

\begin{center}
\begin{tabular}{| c | c c c c c c c c c c c c c c c |}
\hline
$s_1$ & 1&	2&	3&	4&	5&	6&	7&	8&	9&	10&	40&	19&	28&	37&	26\\
\hline
$s_2$ & 11&	12&	13&	14&	15&	16&	17&	18&	19&	20&	30&	9&	38&	27&	36\\
\hline
$s_3$ & 21&	22&	23&	24&	25&	26&	27&	28&	29&	30&	20&	39&	8&	7&	37\\
\hline
$s_4$ & 31&	32&	33&	34&	35&	36&	37&	38&	39&	40&	10&	29&	18&	17&	27\\
\hline
\end{tabular}
\end{center}

It is easy to see that this PCD has $d=4$ and $K=40$. Now we prove that $\alpha=21$.

First, we check the singleton condition for all elements. Imagine we want to check the singleton condition for element $x=10y+z$ where $0 \le y \le 3$ and $1 \le z \le 10$. Then $\ell_{\min}(x)=z$. There are $15$ elements in every sequence in total, so $\ell_{\max}(x) \le 15$. It means that for $z \le 7$ it is true that $\ell_{\min}(x)+\ell_{\max}(x) \le 7 + 15 = 22 \le 21 + 1$ and the singleton condition holds. For $z \ge 8$ it is easy to see that every element occurs exactly twice and $\ell_{\max}(x)=21-z$, so $\ell_{\min}(x)+\ell_{\max}(x)=21 \le 21 + 1$. So the singleton condition also holds.

The hard part is to prove the triplet condition. It is not hard to see that most triplets can be covered by covering their primary positions. For the remaining triplets, some observations may be done that reduce the search space, but in the end, we still would need to go through many cases. Instead of that, we provide a computer program that manually checks the singleton condition for all elements and the triplet condition for all triplets according to the definition. It can be found in \cite{ourRepo}.

\ 

For $d=5$ we present the following ($5$, $40$, $18$) PCD which gives us $\frac{K}{\alpha} = \frac{40}{18} > 2.22222$.

\begin{center}
\begin{tabular}{|c| c c c c c c c c c c c c c |}
\hline
$s_1$ & 1&	2&	3&	4&	5&	6&	7&	8&	24&	31&	38&	30&	14\\
\hline
$s_2$ & 9&	10&	11&	12&	13&	14&	15&	16&	32&	40&	6&	31&	22\\
\hline
$s_3$ & 17&	18&	19&	20&	21&	22&	23&	24&	8&	7&	39&	15&	30\\
\hline
$s_4$ & 25&	26&	27&	28&	29&	30&	31&	32&	40&	16&	23&	39&	6\\
\hline
$s_5$ & 33&	34&	35&	36&	37&	38&	39&	40&	16&	32&	15&	23& \\
\hline
\end{tabular}
\end{center}

It is again easy to see that this PCD has $d=5$ and $K=40$. Now we prove that $\alpha=18$.

First, we check the singleton condition for all elements. Imagine we want to check the singleton condition for element $x=8y+z$ where $0 \le y \le 4$ and $1 \le z \le 8$. Then $\ell_{\min}(x)=z$. There are at most $13$ elements in every sequence in total, so $\ell_{\max}(x) \le 13$. It means that for $z \le 6$ it is true that $\ell_{\min}(x)+\ell_{\max}(x) \le 6 + 13 = 19 \le 18 + 1$ and the singleton condition holds.
Elements with $z=7$ are not located in the level $13$, so for them $\ell_{\max}(x) \le 12$ and $\ell_{\min}(x) + \ell_{\max}(x) \le 7 + 12 = 19 \le 18 + 1$. Elements with $z=8$ are not located in the levels $12$ and $13$, so for them $\ell_{\max}(x) \le 11$ and $\ell_{\min}(x) + \ell_{\max}(x) \le 8 + 11 = 19 \le 18 + 1$.

The hard part again is to prove the triplet condition. We again skip this part and refer to the computer program that does it provided in \cite{ourRepo}.

\end{proof}

We try to answer a question raised in Section \ref{sec:constructions}:

\paragraph*{Do good covering designs always admit a multi-matching?}

Phrasing it more formally, we get the following conjecture:

\begin{conjecture}
    For every integer $d \ge 3$, there exists a bound $\beta_d > 0$ such that there exists at least one covering design with $d$ blocks with $\mathrm{freq} \le \beta_d$ and for all $(v, k, 2)$ covering designs with $d$ blocks and $\mathrm{freq}  \le \beta_d$ their scaled version $(vd, kd, 2)$ admits multi-matching.
\end{conjecture}

In other words, if this conjecture would be true, it would mean that $\gamma_d \ge \frac{d}{3 - 2 / \mathrm{freq}_d}$ where $\mathrm{freq}_d$ is the infimum of frequencies of all covering designs with $d$ blocks. 
We do not know whether this conjecture is true or false, but there are some signs indicating that it may be correct.

But what if this conjecture is false? Do we get any reasonable bounds? It turns out the answer is ``yes''. We can simply create covering designs that are more redundant. Imagine there is a $(vd, kd, 2)$ covering design with $d$ blocks that does not admit multi-matching. Then we can create a $(vd, kd+v, 2)$ covering design (where blocks are multisets instead of sets but it does not bother us because we can leave only the first occurrence of every element in every sequence of our prefix covering design) that admits multi-matching. This can be done by adding unique $v$ values to every block (adding all $vd$ values in total). For example, add $(i-1)v + 1$, $(i-1)v+2$, $\ldots$, $iv$ to the $i$-th block. In this way, these new elements may be chosen as unique elements in each block for multi-matching. If we plug this covering design into our formula, we get a $\frac{d}{3 - 2 \frac{v}{kd + v}}$ lower bound. In Theorem \ref{thm:general_lower_bound} will see that for ``good enough'' covering designs $v = o(kd)$, so it does not change the asymptotic lower bound in terms of $d$ but worsens lower bounds for specific values of $d$.

\begin{proof}[Proof of Theorem \ref{thm:general_lower_bound}]

    Right now, we have some specific lower bounds for exact values of $d$ and some generic construction for arbitrary $d$, which depends on covering design's ``quality''. We would like to get lower bounds that depend only on $d$. To do this, we look at specific covering designs admitting multi-matching. Specifically, finite projective planes. They are $(m^2 + m + 1, m + 1, 2)$ covering designs with $d = m^2 + m + 1$ blocks. They are known to exist for $m=p^n$ where $p$ is a prime number and $n$ is a positive integer. Actually, projective planes are even more powerful than covering designs, they are so-called \textit{balanced incomplete block designs}, which means that every element is located in the same number of blocks, and every pair of elements is located in the same number of blocks (in our case, exactly once). In this situation, every block has size $m + 1$, and every element is located in $m + 1$ blocks. It means that if we look at blocks and elements as a bipartite graph, every vertex in both parts has degree $m + 1$, so this is an $(m+1)$-regular bipartite graph. Due to Hall's theorem, it admits matching, so in every block, there exists a unique element (in this case the number of blocks is equal to the number of elements, so multi-matching is actually just a matching). So finite projective planes satisfy the conditions of Theorem \ref{thm:transformation} and can be used to construct a $\frac{d}{3 - 2 \frac{v}{kd}} = \frac{d}{3 - \frac{2}{k}} = \frac{d}{3} + \frac{d}{3} \cdot \frac{2}{3k-2}$ lower bounds for $d=m^2 + m + 1$ and $m$ being a prime power. This construction gives us a $\frac{d}{3} + \frac{d}{3} \cdot \frac{2}{3k-2} = \frac{d}{3} + \frac{d}{3} \cdot \frac{2}{3m + 1} = \frac{d}{3} + \frac{2}{9} \cdot \frac{d}{m+\frac{1}{3}} = \frac{d}{3} + \frac{2}{9} \sqrt{d} + o\left(\sqrt{d}\right)$ lower bound for these specific values of $d$ (because $\sqrt{d} - 1 < m + \frac{1}{3} < \sqrt{d}$).

Now we generalize this construction for all values of $d$. Remember how we created scaled constructions for covering designs. We had $K=(nk+v')d$ and $\alpha \le 3nk-(2n-3)v'$ where $v'=\frac{v}{d}$ and all elements had primary positions in the first $nk+v'$ levels. To generalize this construction for $d' > d$, we will add $d'-d$ new dimensions, and every one of them will contain $nk-(2n-1)v'$ unique elements. We claim that the value of $\alpha$ is still $\le 3nk-(2n-3)v'$ because if a triplet contains at least one element from these new dimensions, it can be covered by covering primary positions of all three elements: $(nk-(2n-1)v') + (nk+v') + (nk+v') = 3nk-(2n-3)v'$. And $K=(nk+v') \cdot d + (nk-(2n-1)v') \cdot (d'-d)$. For our specific case it means that $\alpha \le 3nk-2n+3$ and $K=(nk+1)d+(nk-2n+1)(d'-d)=(nk+1)d' - 2n(d'-d)$ because $v'=1$. Then $\frac{K}{\alpha} \ge \frac{(nk+1)d' - 2n(d'-d)}{3nk-2n+3} = \frac{(k+\frac{1}{n})d' - 2(d'-d)}{3k-2+\frac{3}{n}}$. It approaches $\frac{kd'-2(d'-d)}{3k-2}$ when $n$ approaches infinity. It is equal to $\frac{d'}{3} \cdot \frac{3k-6 \frac{d'-d}{d}}{3k-2}=\frac{d'}{3}(1 + \frac{2 - 6 \frac{d'-d}{d'}}{3k-2}) = \frac{d'}{3} + \frac{d'}{3} \cdot \frac{2 - 6\frac{d'-d}{d'}}{3k-2}$.
For the fixed value of $d'$ we will take $d \coloneqq p_i^2 + p_i + 1$ where $p_i$ is the largest prime such that $p_i^2 + p_i + 1 \le d'$. The closer $d$, the better bound we will get. On average, the distance between neighboring primes is logarithmic, so $d'-d \le (p_{i+1}^2 + p_{i+1} + 1) - (p_i^2 + p_i + 1) = O(p_{i+1}^2 - p_i^2) = O(p_i \cdot (p_{i+1} - p_i)) = O(\sqrt{d'} \cdot \log d')$ on average but it may not hold in the worst case. Legendre's conjecture states that there is a prime between every two consecutive squares, so that would mean that $p_{i+1}-p_i = O(\sqrt{p_i})$ which means that $d'-d = O\left(d'^{\frac{3}{4}}\right)$.
But Legendre's conjecture, as its name hints, is a conjecture, so we will use a weaker result by Ingham~\cite{Ingham37} that states that $p_{i+1}-p_i = O\left(p_i^{\frac{2}{3}}\right)$. For us it means that $d' - d = O\left(d'^{\frac{5}{6}}\right)$. At the same time we know that $k=m+1=\sqrt{d} + o\left(\sqrt{d}\right) = \sqrt{d'} \pm o\left(\sqrt{d'}\right)$, so we can rewrite the lower bound in terms of $d'$:

\[\frac{d'}{3} + \frac{d'}{3} \cdot \frac{2 - 6\frac{d'-d}{d'}}{3k-2} = \frac{d'}{3} + \frac{d'}{3} \cdot \frac{2 - 6 \frac{O\left(d'^{\frac{5}{6}}\right)}{d'}}{3 \sqrt{d'} \pm o(\sqrt{d'})} = \]
\[= \frac{d'}{3} + \frac{d'}{3} \cdot \frac{2 - o(1)}{(3 \pm o(1)) \sqrt{d'}} = \frac{d'}{3} + \frac{2}{9}\sqrt{d'} - o\left(\sqrt{d'}\right).\]

Which matches the bound we would get if $d'$ admitted finite projective planes in the first two summands.

\end{proof}

\section{Limits~--- Omitted Proofs} \label{sec:limits-appendix}

\begin{proof}[Proof of Scaling Lemma \ref{lemma:scaling}] Let $s_1,\dots, s_d$ be a $(d, K, \alpha)$ PCD and call $\lambda$ the \emph{scaling factor}. We construct a PCD $s'_1,\dots, s'_d$ with $K' \coloneqq \lambda\cdot K$ elements and $\alpha' \coloneqq \lambda\cdot\alpha$ by replacing every element of $s_1,\dots, s_d$ with $\lambda$ elements. Specifically, for every primary position $(i, \ell)$ with value $x$, we replace $x$ by the elements $\lambda x-\lambda+1, \lambda x-\lambda+2, \ldots, \lambda x-1, \lambda x$.
If $(i,\ell)$ is a copy of $x$, we replace $x$ by the elements $\lambda x,  \lambda x-1, \ldots, \lambda x-\lambda+2, \lambda x-\lambda+1$.

We need to check that the triplet and singleton conditions for $\alpha' = \lambda \cdot \alpha$ are satisfied. For every triplet $\{a, b, c\} \in \binom{[K']}{3}$ of elements, consider the triplet $\left\{\left\lceil \frac{a}{\lambda} \right\rceil, \left\lceil \frac{b}{\lambda} \right\rceil, \left\lceil \frac{c}{\lambda} \right\rceil\right\} \subseteq [K]$. \footnote{In the definition of prefix covering designs, the triplet condition should be satisfied for all $\{a, b, c\} \in \binom{[K]}{3}$. We may assume that the triplet condition should be satisfied for all $\{a, b, c\} \subseteq K$ because if some of these three values are equal, we may add arbitrary values to make it a triplet (it is possible because $K \ge 3$) and a coverage for such three distinct elements will cover $\{a, b, c\}$.} Since $s_1,\dots, s_d$ is a prefix covering design, there must exist prefixes $s_i[.. \ell]$, $s_{i'}[.. \ell']$ and $s_{i''}[.. \ell'']$ that cover this triplet and satisfy $\ell + \ell' + \ell'' \le \alpha$. It is easy to see that prefixes $s'_i[.. \lambda\ell]$, $s'_{i'}[.. \lambda\ell']$ and $s'_{i''}[.. \lambda\ell'']$ cover the elements $\{a, b, c\}$ and have a total length of $\lambda\ell + \lambda\ell' + \lambda\ell'' \le \lambda \alpha = \alpha'$.

	The singleton condition follows similarly: Let $a\in [K']$ be an element occurring at least twice in $s'_1,\dots, s'_d$. Define $x \coloneqq \left\lceil \frac{a}{\lambda} \right\rceil$ and $y \coloneqq \lambda x - a$. Consider the primary position $(i, \ellmin(x))$ of the element $x$ in $s_1,\dots, s_d$. Then $s'_i[\lambda \ellmin(x) - y] = a$ by construction of $s'_i$. Thus, $\ell'_{\min}(a) \le \lambda \ellmin(x) - y$.
	On the other hand, for any other position $(i, \ell)$ of $a$ in $s'_1,\dots, s'_d$, we have with $z \coloneqq \left\lceil \frac{\ell}{\lambda} \right\rceil$ that  $s_{i}[z] = x$ and thus $\ell=\lambda (z-1) + y + 1$ by construction of $s'_1,\dots, s'_d$. Since $z \le \ell_{\max}(x)$, we conclude that
	\begin{align*}
		\ellmin'(a) + \ellmax'(a) & \le (\lambda \ellmin(x) - y)+ (\lambda (\ellmax(x)-1) + y + 1) \\
					  & = \lambda(\ellmin(x)+\ellmax(x)-1) + 1 \le \lambda \alpha + 1= \alpha'+1,
	\end{align*} where in the last line we used that $\ell_{\min}(x) + \ell_{\max}(x) \le \alpha + 1$ (since $s_1,\dots, s_d$ is a $(d,K,\alpha)$ PCD). Thus, $s'_1,\dots, s'_d$ is a $(d,\lambda K,\lambda \alpha)$ PCD, as desired.
\end{proof}

\begin{proof}[Proof of Theorem \ref{thrm:no_tight_bound}]

    Assume for contradiction that $\gamma_4 \ge 2$. Then there exists a sequence of $(4, K_i, \alpha_i)$ prefix covering designs such that $\liminf_{i \to \infty} \frac{K_i}{\alpha_i} \ge 2$. By the Scaling Lemma \ref{lemma:scaling} we may scale every prefix covering design in the sequence to have larger value of $K_i$ than the previous one, so we may assume that $K_i < K_{i+1}$ for all $i \ge 1$. We prove a series of properties that hold for all PCDs in this sequence for sufficiently large $i$. For clarity of presentation (by a slight abuse of notation), we will use asymptotic statements (over $i$) and write $\alpha$ instead of $\alpha_i$ and $K$ instead of $K_i$. For example, we can already write that $K \ge 2 \alpha - o(\alpha)$ and $\alpha \le \frac{K}{2} + o(K)$. For any considered PCD $s_1, s_2, \ldots, s_d$ we may also assume that no elements are located in the same sequence $s_i$ twice because otherwise we could simply delete the second occurance, and it could only improve the construction.

\begin{lemma}\label{lemma:small_primary_positions}
    If $\alpha \le \frac{K}{2} + o(K)$, the largest level of a primary position is at most $\frac{\alpha}{2} + o(\alpha)$.
\end{lemma}

\begin{proof}[Proof of Lemma \ref{lemma:small_primary_positions}]
    
    Call $(i_m, m)$ the largest primary position (with largest $m$) of this PCD. We will prove that $m \le \frac{\alpha}{2} + o(\alpha)$. Assume for contradiction that $m = \frac{\alpha}{2} + \varepsilon \alpha$ for some $\varepsilon > \varepsilon_0 > 0$ where $\varepsilon_0$ is a constant independent of $K$ and $\alpha$. It is easy to see that $s_{i_m}[m]$ cannot have copies because otherwise, we would have $m + m' \ge m + m \ge \alpha + 2 \varepsilon \alpha > \alpha$ where $m'$ is the level of any copy, contradicting the singleton condition for $s_{i_m}[m]$. We prove that prefix $s_{i_m}[.. m]$ and prefixes $s_i[.. \frac{\alpha}{2} - \varepsilon \alpha]$ for all $i \neq i_m$ cover all $K$ elements. Assume for contradiction that it is not true. Distinguish two cases of where else elements can be located.

\begin{enumerate}
    \item If some element is located at a position $(i_m, y)$ for $y \ge m + 1 \ge \frac{\alpha}{2} + \varepsilon \alpha + 1$ then it cannot be its primary position because $m$ was the largest primary position, so this element should have a primary position in some other sequence $i \neq i_m$ at level at most $\frac{\alpha}{2} - \varepsilon \alpha$ for the singleton condition to hold. But all such elements are covered by prefixes $s_{i}[.. \frac{\alpha}{2} - \varepsilon \alpha]$.

    \item If an element is not located in sequence $i_m$, define its primary position as $(i_y, y)$. But then we need prefixes of total size at least $m + y$ to cover any triplet of form $\{s_{i_m}[m], s_{i_y}[y], c\}$ with $c \in [K]$. So $m + y \le \alpha$ which yields $y \le \frac{\alpha}{2} - \varepsilon \alpha$, and this element is covered by prefix $s_{i_y}[.. \frac{\alpha}{2} - \varepsilon \alpha]$.
\end{enumerate}

    So indeed these prefixes cover all elements. Thus, they must contain at least $K$ elements. But their total size is $m + 3 \cdot \left(\frac{\alpha}{2} - \varepsilon \alpha\right) = \frac{\alpha}{2} + \varepsilon \alpha + 3 \cdot \left(\frac{\alpha}{2} - \varepsilon \alpha\right) = (2 - 2 \varepsilon) \alpha < (2 - 2 \varepsilon_0) \alpha < K$ for big enough $K$. So we get a contradiction.

\end{proof}

Consequently, all primary positions are located on the first $m = \frac{\alpha}{2} \pm o(\alpha)$ levels (it is obvious that $m \ge \frac{\alpha}{2} - o(\alpha)$ as there are $K \ge 2 \alpha - o(\alpha)$ elements that should have primary positions in $4$ sequences). But these levels contain $K + o(K)$ elements in total, so almost all\footnote{Here and later, \emph{almost all} is understood as all but $o(K)$.} elements at these levels are primary. 

\begin{lemma}\label{lemma:existance_of_copies}
    Let $m$ denote the largest level of a primary position. There exists $\varepsilon > 0$ independent of $K$ such that every element with primary position level $p$ with $(1 - \varepsilon)m \le p \le m$ has at least one copy.
\end{lemma}

\begin{proof}[Proof of Lemma \ref{lemma:existance_of_copies}]
    Assume for contradiction that there exists an element with a primary position level between $m - \varepsilon m$ and $m$ that does not have any copies. Define its primary position, without loss of generality, as $(1, m - x)$ where $0 \le x \le \varepsilon m$. We distinguish two cases whether there exist any element with a primary position in the levels $[m - 2 \varepsilon m, m]$ from sequences $2$, $3$ or $4$ that does not have a copy or not.

\begin{enumerate}
    \item If it does exist, define its primary position without loss of generality as $(2, m - y)$ where $0 \le y \le 2 \varepsilon m$. Consider how to cover triplets of the form $\{s_1[m - x], s_2[m - y], s_i[z]\}$ for $i \in \{3, 4\}$ and $z \in [m- 2 \varepsilon m, m]$ such that $(i, z)$ is the primary position of $s_i[z]$ (it exists as there are $4m = K + o(K)$ positions on the first $m$ levels, so almost all of them are primary). We cannot cover them with three sequences because it would take $\ge (3 - 5 \varepsilon) m \ge \left(\frac{3-5 \varepsilon}{2}-o(1)\right) \alpha > \alpha$ for small enough $\varepsilon$. Thus, we must cover such triplets with at most two sequences. But $s_1[m-x]$ and $s_2[m - y]$ do not have copies, so $s_i[z]$ must have some copy in sequence $1$ or $2$. If we copy it to sequence $1$, its position must be $\le m + 2 \varepsilon m + o(m)$ to comply with the triplet condition, and if we copy it to sequence $2$, its position must be $\le m + \varepsilon m + o(m)$ to comply with the triplet condition. As we have $o(m)$ copy positions in the first $m$ levels, it means that there are $3 \varepsilon m + o(m)$ suitable copy positions in the first two sequences in total. But there are $4 \varepsilon m - o(m)$ primary elements from layers $[m - 2 \varepsilon m, m]$ in sequences $3$ and $4$ that should be copied. This yields a contradiction.

    \item If there does not exist such an element $s_2[m - y]$, then all elements with primary positions on levels $[m - 2 \varepsilon m, m]$ in sequences $2$, $3$ and $4$ have copies. If some of them have copies in sequences $2$, $3$ or $4$ on positions $\ge m + \varepsilon m + \delta m$ for some constant $\delta > 0$ independent of $K$, these positions cannot be used to cover triplets containing $s_1[m - x]$ due to the triplet condition, so this situation is equivalent to the situation where such an element does not have a copy at all, and we understood that it is impossible. So only levels $\le m + \varepsilon m + o(m)$ are available in the sequences $2$, $3$ and $4$ for copies. At the same time, these elements have primary positions $\ge m - 2 \varepsilon m$, so due to the singleton condition their copies in sequence $1$ cannot be higher than level $m + 2 \varepsilon m + o(m)$. As we know that there are $o(m)$ copy positions on the first $m$ levels of all sequences, there are $2 \varepsilon m + \varepsilon m + \varepsilon m + \varepsilon m + o(m) = 5 \varepsilon m + o(m)$ copy positions for these elements. But there are $6 \varepsilon m - o(m)$ such primary elements from sequences $2$, $3$ and $4$ in total. Consequently, some of them cannot have copies, yielding a contradiction.
\end{enumerate}

\end{proof}

\begin{lemma}\label{lemma:uniqueness_of_copies}
    Let $m$ denote the largest level of a primary position. There exists $\varepsilon > 0$ independent of $K$ such that \emph{almost} every element with primary position level $p$ with $(1 - \varepsilon)m \le p \le m$ has \emph{exactly} one copy and this copy is located on level $2m-p \pm o(m)$.
\end{lemma}

\begin{proof}[Proof of Lemma \ref{lemma:uniqueness_of_copies}]
    From Lemma \ref{lemma:existance_of_copies}, we know that all of these elements with primary positions from levels $[m - \varepsilon m, m]$ have at least one copy. At the same time, due to the singleton condition, their copies should be on levels $\le m + \varepsilon m + o(m)$, and as we know, there are $o(m)$ copy positions in the first $m$ levels, so there are $4 \varepsilon m + o(m)$ copy positions for them in total. But there are $4 \varepsilon m - o(m)$ such elements. Thus, almost all elements with primary positions from levels $[m - \varepsilon m, m]$ have \emph{exactly} one copy. For an element with primary position $(k, m - x)$, such a copy should be on the level $\le m + x + o(m)$ due to the singleton condition. At the same time, its level can not be much smaller. Assume for contradiction that there are $\ge \delta m$ elements with primary positions $(k_i, m - x_i)$ for some constant $\delta > 0$ independent of $K$, and their copies are located at positions $y_i \le m + x_i - \gamma m$ for some constant $\gamma > 0$ independent of $K$. Define $S$ as the sum of all copy positions of all elements with primary positions from levels $[m - \varepsilon m, m]$ that have exactly one copy. Every such position is at most $m + x + o(m)$ for an element with a primary position on level $m-x$, but there are $\ge \delta m$ elements that have copy positions $\le m + x - \gamma m$. It means that $S \le 4 \left(\sum_{i=0}^{\varepsilon m} m+i \right) - (\gamma m) \cdot (\delta m) = (4 \varepsilon + 2 \varepsilon^2 + o(1) - \gamma \cdot \delta) m^2$. On the other hand, we can find a lower bound on $S$. Only $o(m)$ elements have copy positions smaller than $m$, so there are $4 \varepsilon m - o(m)$ elements that should have copies on levels starting from the $m$-th. If we order these copy positions in the increasing order, the first four positions are $\ge m$, the second four positions are $\ge m + 1$, and so on: positions with indices $4i+1$, $4i+2$, $4i+3$ and $4i+4$ are $\ge m + i$ (because otherwise, we would need to pack at least $4i+1$ elements on $i$ levels of four sequences). It means that their sum $S$ is at least $4 \varepsilon m^2 + 2 \varepsilon^2 m^2 - o(m^2) = (4 \varepsilon + 2 \varepsilon^2 - o(1)) m^2$. Which contradicts the upper bound we acquired.

\end{proof}

So almost all of these elements with primary positions from levels $[m - \varepsilon m, m]$ have exactly one copy, and if the primary position of such an element is $m - x$, then its copy is located on level $m + x \pm o(m)$ for almost all elements. It also means that for constants $0 \le \varepsilon_2 < \varepsilon_1 \le \varepsilon$ independent of $K$ almost all elements from $[m - \varepsilon_1 m, m - \varepsilon_2 m)$ are primary, and almost all of them have exactly one copy each, and for almost all of them, these copies are located on levels $[m + \varepsilon_2 m - o(m), m + \varepsilon_1 m + o(m))$. But there are $o(m)$ positions on levels $[m + \varepsilon_2 m - o(m), m + \varepsilon_1 m + o(m)) \setminus [m + \varepsilon_2 m, m + \varepsilon_1 m)$, so we can say that almost all of these elements have their copies on levels $[m + \varepsilon_2 m, m + \varepsilon_1 m)$.

    Now all the preparatory work is done, and we are ready to prove that such constructions are not possible. We say that an element with primary position $(i, m - x)$ ``is uniquely copied'' to sequence $j$ if it has exactly one copy and it is located on level $m + x \pm o(m)$ in the sequence $j$. We fix some sufficiently small constant $\gamma > 0$ independent of $K$. We say that elements from levels $[m - \varepsilon_1 m, m - \varepsilon_2 m)$ ($0 \le \varepsilon_2 < \varepsilon_1 \le \varepsilon$) from sequence $i$ ``are uniquely copied'' to sequence $j$ if there exist $\ge \gamma m$ elements with primary positions on these levels in sequence $i$ that have their one and only copy in sequence $j$ on levels $[m + \varepsilon_2 m, m + \varepsilon_1 m)$.

    Call $P_{[\ell, r)}$ the set of all pairs $(i, j)$ such that there are $\ge \gamma m$ elements from levels $[\ell, r)$ from sequence $i$ that are uniquely copied to levels $[2m-r, 2m-\ell)$ in sequence $j$.

\begin{lemma} \label{lm:five-division}
        There exists $\delta > 0$ independent of $K$ such that there exist $0 < \varepsilon_2 < \varepsilon_1 \le \varepsilon$ such that $B \coloneqq \frac{\varepsilon_1m - \varepsilon_2m}{5} \ge \delta m$ and $P_{[m - \varepsilon_1 m, m - \varepsilon_2 m)} = P_{[m - \varepsilon_1 m + iB, m - \varepsilon_1 m + (i+1)B)}$ for all $0 \le i \le 4$.
\end{lemma}

\begin{proof}
    We prove this statement by a construction. $\delta$ will be equal to $\frac{\varepsilon}{5^{12}}$ (and we should take $\gamma$ much smaller than $\delta$). We provide an iterative process with at most $12$ steps that will lead us to desired values of $\varepsilon_1$ and $\varepsilon_2$.

    We start with a segment $[m - \varepsilon m, m)$. On every iteration we divide our current segment $[\ell, r)$ into five almost equal parts and consider $P$-sets for these six segments. Define them as $P$, $P_1$, $P_2$, $P_3$, $P_4$ and $P_5$. If all of them are equal to each other, the lemma holds and we use the current segment $[\ell, r)$ as the answer. Otherwise there exists $1 \le i \le 5$ such that $P \neq P_i$. As all five subsegments of $[\ell, r)$ lie strictly inside $[\ell, r)$, this implies that $P_i \subsetneq P$ and thus $|P_i| \le |P| - 1$. We replace our current segment $[\ell, r)$ with $[\ell + (i-1)(r-\ell) / 5, \ell + i(r - \ell) / 5)$ and continue to the next stage. Note that $|P_{[m - \varepsilon m, m)}| \le 4 \cdot 3 = 12$ and on every next iteration the size of the segment decreases by a factor of $5$ and the size of $P$-set decreases by at least one. Consequently, after at most $12$ iterations, we will get a suitable segment of size at least $\frac{\varepsilon}{5^{12}}$.

\end{proof}

Due to Lemma \ref{lm:five-division} we kow that there exists a segment $[A, A + 5B)$ such that $A \ge m - \varepsilon m$, $A + 5B \le m$, $B = \Theta(m)$, and this segment satisfies the lemma statement. Now we answer a question: which pairs $(i, j)$ can be in the $P$-set of the segment $[A, A + 5B)$?

\begin{lemma}\label{lm:first-forbidden-edge-triple}
        If $1 \le i, j, k, \ell \le 4$ are four distinct integers, at least one of pairs $(i, j)$, $(j, i)$ and $(k, \ell)$ is not present in the $P$-set of $[A, A + 5B)$.
\end{lemma}

\begin{proof}
	Assume for contradiction that all these pairs are present in the $P$-set of $[A, A + 5B)$. Then all these pairs are contained in $P$-sets of all $5$ subsegments due to Lemma \ref{lm:five-division}. In particular, it means that there is an element with primary position $(i, A + x)$ where $0 \le x < B$ that has its only copy in sequence $j$ on level $2m - A - x \pm o(m)$. Also, there exists an element with primary position $(j, A + y)$ where $0 \le y < B$ that has its only copy in sequence $i$ on level $2m - A - y \pm o(m)$, and an element with primary position $(k, A + 4B + z)$ where $0 \le z < B$ that has its only copy in sequence $\ell$ on level $2m - A - 4B - z \pm o(m)$. We claim that it is impossible to cover these three elements with prefixes of total size at most $\alpha = 2m + o(m)$. If we use three different prefixes, the sum of their lengths is at least $3 \cdot (1 - \varepsilon)m > \alpha$ because $A \ge m - \varepsilon m$. Element $s_{k}[A + 4B + z]$ does not lie in the same sequences as two other elements, so it must be covered alone. It has $\ell_{\min}(s_{k}[A + 4B + z]) = A + 4B + z$, so such prefix must have size at least $A + 4B + z$. The two other elements are both located in sequences $i$ and $j$, so we should take a prefix in one of them. Without loss of generality, say we take sequence $j$. The coverage will be of size $(A + 4B + z) + (2m - A - x \pm o(m)) = 2m + 4B + z - x \pm o(m) \ge 2m + 4B - x \pm o(m) \ge 2m + 3B \pm o(m) = (2 + \Theta(1))m > \alpha$. So we cannot have all three pairs $(i, j)$, $(j, i)$ and $(k, \ell)$ in our $P$-set at the same time.
\end{proof}

    We prove the same statement about other three pairs.

\begin{lemma}\label{lm:second-forbidden-edge-triple}
        If $1 \le i, j, k, \ell \le 4$ are four distinct integers, at least one of pairs $(i, j)$, $(j, k)$ and $(k, \ell)$ is not present in the $P$-set of $[A, A + 5B)$.
\end{lemma}

\begin{proof}
   Assume for a contradiction that all these pairs are present in the $P$-set of $[A, A + 5B)$. Then all these pairs are contained in $P$-sets of all $5$ subsegments due to Lemma \ref{lm:five-division}. In particular, there is an element with primary position $(i, A + 2B + x)$ where $0 \le x < B$ that has its only copy in sequence $j$ on level $2m - A - 2B - x \pm o(m)$. Also there exists an element with primary position $(j, A + y)$ where $0 \le y < B$ that has its only copy in sequence $k$ on level $2m - A - y \pm o(m)$ and an element with primary position $(k, A + 4B + z)$ where $0 \le z < B$ that has its only copy in sequence $l$ on level $2m - A - 4B - z \pm o(m)$. We claim that it is impossible to cover these three elements with prefixes of total size at most $\alpha = 2m + o(m)$. If we use three different prefixes, the sum of their lengths is at least $3 \cdot (1 - \varepsilon)m > 2m + o(m)$ because $A \ge m - \varepsilon m$. Now imagine that we use two sequences (obviously one sequence is impossible because there is no one sequence where all three elements are located). It is easy to see that we should take a prefix in sequence $k$ because otherwise, element $s_{k}[A + 4B + z]$ can be covered only using sequence $\ell$, and no other elements from these three are located in that sequence and level on which $s_{k}[A + 4B + z]$ is located in that sequence is not smaller than in sequence $k$ for this element (because in sequence $k$ it has its primary position), so it would be not worse to take prefix in sequence $k$ instead. So we should take sequence $k$ and either sequence $i$ or $j$. If we take sequence $i$, the coverage will be of size $(A + 2B + x) + (2m - A - y \pm o(m)) = 2m + 2B + x - y \pm o(m) \ge 2m + 2B - y \pm o(m) \ge 2m + B \pm o(m) = (2 + \Theta(1))m > \alpha$ because $B = \Theta(m)$. If we take sequence $j$, the coverage will be of size $(A + 4B + z) + (2m - A - 2B - x \pm o(m)) = 2m + 2B + z - x \pm o(m) \ge 2m + 2B - x \pm o(m) \ge 2m + B \pm o(m) = (2 + \Theta(1))m > \alpha$. So we cannot have all three pairs $(i, j)$, $(j, k)$ and $(k, \ell)$ in our $P$-set at the same time.
\end{proof}

These two limitations are sufficient for us to prove that such construction cannot exist. Consider a directed graph $G = (V, E)$ where $V = [4]$ and $E = P$. The weight of an edge $i \to j$ is equal to the number of elements going from sequence $i$ to sequence $j$ (notice that due to the definition of $P$-set all weights are $\ge \gamma m$). In every sequence there are $5B \pm o(m)$ elements that should be uniquely copied somewhere, and there are $5B$ places (on levels $[2m-A-5B, 2m-A)$) where elements from other sequences can be uniquely copied to in this sequence.
It means that for every vertex the sum of all outgoing edges is $5B \pm o(m)$ and the sum of all ingoing edges is $5B \pm o(m)$.

    Distinguish two cases whether there exist pairs $(i, j)$ and $(j, i)$ for some $i \neq j$ in the $P$-set or not.

\begin{enumerate}
    \item First, consider the case where such two pairs do not exist. So if there exists a pair $(i, j)$, the reversed pair $(j, i)$ cannot be present. Without loss of generality, say that there exists an edge $1 \to 2$. It means that there cannot exist an edge $2 \to 1$. Without loss of generality, say that there exists an edge $2 \to 3$. It means that there cannot exist edge $3 \to 2$. Elements from sequence $3$ cannot be uniquely copied to sequence $4$, because otherwise we get a construction where we have pairs $(i, j)$, $(j, k)$ and $(k, \ell)$ where $i=1$, $j=2$, $k=3$ and $\ell=4$ which would violate Lemma \ref{lm:second-forbidden-edge-triple}. So the only place where elements from sequence $3$ can be uniquely copied to is sequence $1$. Now we have a triangle $1 \to 2 \to 3 \to 1$. In the same, way we can see that there cannot exist edges $1 \to 4$ and $2 \to 4$ because otherwise, we would also violate Lemma \ref{lm:second-forbidden-edge-triple}. So it means that there are no edges into vertex $4$. It means that there can exist only $o(m)$ elements that are uniquely copied there. But it should have at least $5B \pm o(m) = \Theta(m)$ ingoing elements. It is a contradiction. See Figure \ref{fig:impossible-graph-1} for more understanding.

\begin{figure}
\begin{center}        
\begin{tikzpicture}[main/.style = {draw, circle}, scale=2]

    \node[main] at (0, 0) (3) {$3$};
    \node[main] at (-0.5,0.866) (1) {$1$};
    \node[main] at (0.5,0.866) (2) {$2$};
    \node[main] at (0, -1) (4) {$4$};

    \draw[line width=0.5mm, ->, color=green!40!gray] (1) to [out=20,in=160] (2);
\draw[line width=0.5mm, ->, color=green!40!gray] (2) to [out=260,in=40] (3);
\draw[line width=0.5mm, ->, color=green!40!gray] (3) to [out=140,in=280] (1);
\draw[line width=0.5mm, dashed, ->, color=red] (2) to [out=200,in=340] (1);
\draw[line width=0.5mm, dashed, ->, color=red] (3) to [out=80,in=220] (2);
\draw[line width=0.5mm, dashed, ->, color=red] (1) to [out=320,in=100] (3);
\draw[line width=0.5mm, dashed, ->, color=red] (3) to [out=270,in=90] (4);
\draw[line width=0.5mm, dashed, ->, color=red] (1) to [out=250,in=150] (4);
\draw[line width=0.5mm, dashed, ->, color=red] (2) to [out=290,in=30] (4);

\end{tikzpicture}
\end{center}
\caption{Solid edges are present, dashed edges cannot be present. We get a contradiction because no ingoing edges to vertex $4$ are allowed.}
\label{fig:impossible-graph-1}
\end{figure}
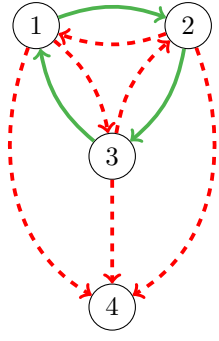

    \item Second, consider the case where there exist both pairs $(i, j)$ and $(j, i)$ for some $i \neq j$. Define $k, \ell$ such that $\{k, \ell\} = \{1, 2, 3, 4\} \setminus \{i, j\}$. There cannot exist edges $(k, \ell)$ or $(\ell, k)$ because it would violate Lemma \ref{lm:first-forbidden-edge-triple}. So then it means that almost all elements from sequences $k$ and $\ell$ are uniquely copied to sequences $i$ and $j$. There are $10B \pm o(m)$ such elements, so they take $10B \pm o(m)$ places. But there are $10B \pm o(m)$ places for copy in sequences $i$ and $j$ in total. Then edges $(i, j)$ and $(j, i)$ should have weights $o(m)$. But we cannot have edges with such small values because otherwise, we would not even add such pairs to $P$. This is a contradiction. See Figure \ref{fig:impossible-graph-2} for more understanding.

    \begin{figure}
\begin{center}        
\begin{tikzpicture}[main/.style = {draw, circle}, scale=2]

    \node[main] at (0,0) (i) {$i$};
    \node[main] at (1,0) (j) {$j$};
    \node[main] at (0,-1) (k) {$k$};
    \node[main] at (1, -1) (l) {$\ell$};

\draw[line width=0.5mm, ->, color=green!40!gray] (i) to [out=20,in=160] (j);
\draw[line width=0.5mm, ->, color=green!40!gray] (j) to [out=200,in=340] (i);
\draw[line width=0.5mm, ->, decorate, decoration={snake, amplitude=.4mm}, color=yellow] (k) to [out=90,in=270] (i);
\draw[line width=0.5mm, ->, decorate, decoration={snake, amplitude=.4mm}, color=yellow] (l) to [out=90,in=270] (j);
\draw[line width=0.5mm, ->, decorate, decoration={snake, amplitude=.4mm}, color=yellow] (k) to [out=45,in=225] (j);
\draw[line width=0.5mm, ->, decorate, decoration={snake, amplitude=.4mm}, color=yellow] (l) to [out=135,in=315] (i);
\draw[line width=0.5mm, dashed, ->, color=red] (k) to [out=20,in=160] (l);
\draw[line width=0.5mm, dashed, ->, color=red] (l) to [out=200,in=340] (k);

\end{tikzpicture}
\end{center}
\caption{Solid edges are present, snake edges may be present, dashed edges cannot be present. We get a contradiction because too many elements are uniquely copied to sequences $i$ and $j$.}
\label{fig:impossible-graph-2}
\end{figure}
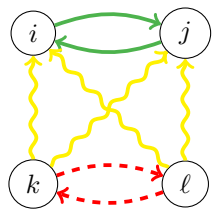

\end{enumerate}

    It concludes the proof of the theorem. It means that for any $(4, K, \alpha)$ PCD, it is true that $\frac{K}{\alpha} \le 2 - \varepsilon'$ for some fixed $\varepsilon' > 0$ independent of $K$ and, for example, no matching lower bound for $4$-dimensional Klee's measure problem is possible using this approach.
\end{proof}

\end{document}